\documentclass[preprint,12pt]{elsarticle}

\usepackage{amssymb}
\usepackage{amsmath}
\usepackage{mathtools}

\usepackage{graphicx,color}
\usepackage{amsthm}
\usepackage{aliascnt}
\usepackage[colorlinks]{hyperref}
\usepackage{framed,ascmac}
\usepackage[linesnumbered,ruled]{algorithm2e}
\usepackage{multicol,multirow}
\usepackage{cite}
\usepackage{tikz}
\usetikzlibrary{arrows,positioning,automata,calc}
\usetikzlibrary{arrows.meta}
\tikzset{
  implication/.style={
    -{Implies[length=4mm,width=3mm]}
  }
}
\usepackage[nameinlink,capitalize,noabbrev]{cleveref}
\usepackage{mathrsfs}
\usepackage[all]{xy}
\usepackage{comment}

\usepackage{graphicx}

\renewcommand{\Join}{%
  \mathop{\scalebox{1.15}{$\bowtie$}}\displaylimits}

\newcommand{\supp}{\mathop{\rm supp}}

\newcommand{\cl}{{\rm cl}}
\newcommand{\conv}{{\rm conv}}

\newcommand{\proj}{\pi}
\newcommand{\avgd}{m^{\to}}
\newcommand{\median}{{\rm median}}

\renewcommand{\mid}{\,:\,}

\newtheorem{theorem}{Theorem}[section]
\newaliascnt{lemma}{theorem}
\newtheorem{lemma}[lemma]{Lemma}
\aliascntresetthe{lemma}

\newaliascnt{proposition}{theorem}
\newtheorem{proposition}[proposition]{Proposition}
\aliascntresetthe{proposition}

\newaliascnt{corollary}{theorem}
\newtheorem{corollary}[corollary]{Corollary}
\aliascntresetthe{corollary}

\newaliascnt{definition}{theorem}
\newtheorem{definition}[definition]{Definition}
\aliascntresetthe{definition}

\newaliascnt{example}{theorem}
\newtheorem{example}[example]{Example}
\aliascntresetthe{example}

\newaliascnt{remark}{theorem}
\newtheorem{remark}[remark]{Remark}
\aliascntresetthe{remark}

\newaliascnt{claim}{theorem}

\aliascntresetthe{claim}

\crefname{theorem}{Theorem}{Theorems}
\Crefname{theorem}{Theorem}{Theorems}

\crefname{lemma}{Lemma}{Lemmas}
\Crefname{lemma}{Lemma}{Lemmas}

\crefname{proposition}{Proposition}{Propositions}
\Crefname{proposition}{Proposition}{Propositions}

\crefname{corollary}{Corollary}{Corollaries}
\Crefname{corollary}{Corollary}{Corollaries}

\crefname{definition}{Definition}{Definitions}
\Crefname{definition}{Definition}{Definitions}

\crefname{example}{Example}{Examples}
\Crefname{example}{Example}{Examples}

\crefname{remark}{Remark}{Remarks}
\Crefname{remark}{Remark}{Remarks}

\crefname{claim}{Claim}{Claims}
\Crefname{claim}{Claim}{Claims}

\crefname{section}{Section}{Sections}
\Crefname{section}{Section}{Sections}
\crefname{subsection}{Subsection}{Subsections}
\Crefname{subsection}{Subsection}{Subsections}
\crefname{equation}{Equation}{Equations}
\Crefname{equation}{Equation}{Equations}
\crefname{algorithm}{Algorithm}{Algorithms}
\Crefname{algorithm}{Algorithm}{Algorithms}
\crefname{appendix}{Appendix}{Appendices}
\Crefname{appendix}{Appendix}{Appendices}

\journal{Journal of Combinatorial Theory, Series A}

\begin{document}

\begin{frontmatter}



\title{UTVPI-representable integer point sets:
discrete convexity, polymorphisms, and pairwise closure}


\author[inst1]{Kei Kimura\texorpdfstring{\corref{cor1}}{}}\ead{kimura.kei.993@m.kyushu-u.ac.jp}
\author[inst2]{Kazuhisa Makino}
\author[inst3]{Shota Yamada}
\author[inst1,inst4]{Ryo Yoshizumi}

\affiliation[inst1]{organization={Kyushu University},
            city={Fukuoka},
            country={Japan}}
\affiliation[inst2]{organization={Kyoto University},
            city={Kyoto},
            country={Japan}}
\affiliation[inst3]{organization={National Institute of Advanced Industrial Science and Technology},
            city={Tokyo},
            country={Japan}}
\affiliation[inst4]{organization={Present address: NTT Social Informatics Laboratories},
            city={Kanagawa},
            country={Japan}}

\cortext[cor1]{Corresponding author.}

\begin{abstract}
We study subsets of the integer lattice represented by single-variable-per-inequality (SVPI), difference-constraint (DC), unit two-variable-per-inequality (UTVPI), and two-variable-per-inequality (TVPI) systems. We relate five viewpoints: inequality representation, discrete convexity, polymorphisms, reconstruction from two-coordinate projections, and fixed points of closure operators.

  Our central result completely characterizes UTVPI-representability. For every set $S\subseteq\mathbb Z^n$ with $n>1$, \[ \begin{aligned} &S\text{ is UTVPI-representable}\\ &\;\Longleftrightarrow\; S\text{ is closed under the directed midpoint and median operations}\\ &\;\Longleftrightarrow\; S\text{ is integrally convex and $2$-decomposable}. \end{aligned} \] The median condition may instead be replaced by closedness under some majority operation, and the same class is the fixed-point class of a pairwise directed-midpoint closure operator. Thus, all five viewpoints yield equivalent characterizations of UTVPI-representability. In particular, $2$-decomposability is exactly the global condition needed to lift the known two-dimensional equivalence between integral convexity and UTVPI-representability to arbitrary dimension.
  
  This theorem is embedded in a broader pairwise-closure theory. For a family $F$ of operations, we define a closure operator by closing every two-coordinate projection under $F$ and joining the resulting sets. Its fixed points are precisely the sets that are both $2$-decomposable and $F$-closed, and we establish a local-to-global criterion for such characterizations. A closed-convex-hull analogue characterizes TVPI-representability. We also characterize SVPI-representability by natural multioperations, prove limitations of operation-based characterizations for several related classes, and determine the complete inclusion hierarchies in the general, Boolean, and two-dimensional settings.

\end{abstract}


\begin{keyword}
UTVPI system \sep integral convexity \sep polymorphism \sep
$2$-decomposability \sep pairwise closure \sep TVPI system
\end{keyword}

\end{frontmatter}


\section{Introduction}
\label{sec:introduction}

A subset of the integer lattice can be studied from five complementary
structural viewpoints: inequality representation, preservation under
coordinatewise operations, discrete convexity, reconstruction from
two-coordinate projections, and fixed-point properties under closure
operators. These viewpoints originate in different areas and are
usually studied separately. Our central result brings them together by
giving equivalent characterizations of integer sets representable by
unit two-variable-per-inequality (UTVPI) systems. This unification is
of independent mathematical interest, as it reveals a common structure
underlying these apparently different descriptions. We also develop a
common framework that explains which aspects of this correspondence
extend to related inequality classes and where such extensions fail.

Among these viewpoints, the interaction between representation and
preservation is particularly well developed. It is fundamental to the
theories of Boolean satisfiability and constraint satisfaction, and
has also proved useful in constraint-based knowledge representation
and reasoning, particularly in qualitative temporal and spatial
reasoning~\citep{BK10,BJMM18,BK21}. A basic example is provided by
Boolean relations: Horn and bijunctive relations are characterized by
closedness under the coordinatewise minimum and majority operations,
respectively~\citep{McK43,Hor51,Sch78}. More generally, polymorphisms
provide an algebraic language for describing solution sets and play a
central role in the structural and algorithmic theory of constraint
satisfaction problems. In this paper, we investigate how this
operation-based perspective interacts with the other four viewpoints
for subsets of the integer lattice represented by systems of linear
inequalities.

The viewpoint of reconstruction from two-coordinate projections is
formalized by the notion of \emph{$2$-decomposability}, whose precise
definition is given in \Cref{subsec:Properties-solutions}. Informally,
a set is $2$-decomposable if membership in the set is completely
determined by its two-coordinate projections. Thus,
$2$-decomposability expresses a local-to-global principle based on
pairwise information. In finite-domain CSPs, closedness under a
majority operation implies $2$-decomposability; more generally,
near-unanimity polymorphisms imply bounded
decomposability~\citep{JCC98}. These results suggest a close connection
between pairwise reconstruction and preservation under coordinatewise
operations, and motivate our unified study of these properties together
with discrete convexity, inequality representability, and fixed-point
characterizations by closure operators.

Against this background, we focus on four classes of systems of linear
inequalities: single-variable-per-inequality (SVPI),
difference-constraint (DC), UTVPI, and
two-variable-per-inequality (TVPI) systems. Their definitions are given
in \Cref{subsec:Polyhedral-representation}. Informally, an SVPI
inequality has at most one nonzero coefficient; a DC inequality has
all its coefficients in $\{0,-1,+1\}$, with at most one coefficient
equal to $+1$ and at most one coefficient equal to $-1$; a UTVPI
inequality has at most two nonzero coefficients, each belonging to
$\{0,-1,+1\}$; and a TVPI inequality has at most two nonzero
coefficients, which may take arbitrary values. These four
representability classes are nested in the order listed above. Among
them, DC, UTVPI, and TVPI systems have been extensively studied in
integer programming and combinatorial
optimization~\citep{Sch91,JMS94,Sub04,LaM05,SuW17,DMZ18,KN24}.

Several known results provide individual pieces of the picture.
DC-representability is characterized by closedness under both the upper
and lower midpoint operations~\citep{Murota03}; in two dimensions,
UTVPI-representability is equivalent to integral
convexity~\citep{MMTT19}; every TVPI-representable set is
median-closed~\citep{BoM18b}; integral convexity is preserved under
coordinate projections~\citep{MoM19}; and closedness under a majority
operation implies $2$-decomposability~\citep{JCC98}. 
These results, however, do not provide a complete higher-dimensional
characterization of UTVPI-representability. To the best of our
knowledge, the precise global condition needed to extend the
two-dimensional characterization to arbitrary dimension had not
previously been identified.

\subsection*{Main results}
Let $\mathbb Z$ and $\mathbb R$ denote the sets of integers and real
numbers, respectively. Throughout this paper, let $n$ be an integer
greater than $1$, and let 
$
V=\{1,\ldots,n\}.
$
Unless otherwise specified, $S$ denotes a subset of $\mathbb Z^V$.

The main results of this paper can be grouped into five interrelated
parts. Among them, the characterization theorem for
UTVPI-representability in \textup{(b)} is the central result of the
paper, while \textup{(a)} provides the general pairwise-closure
framework in which its fixed-point formulation naturally arises.

\medskip
\noindent\textup{(a)}
We develop a general pairwise-closure framework, characterize the
fixed points of the resulting closure operators, and establish a
local-to-global criterion for operation-closed classes. We also
develop a closed-convex-hull analogue that yields a fixed-point
characterization of TVPI-representability.
 
\medskip
\noindent\textup{(b)}
We establish equivalent characterizations of UTVPI-representability
in terms of preservation under coordinatewise operations, integral
convexity together with $2$-decomposability, and a fixed-point
property.

\medskip
\noindent\textup{(c)}
We show that SVPI-representability cannot be characterized by ordinary
operations and instead characterize it using two natural binary multioperations. 

\medskip
\noindent\textup{(d)}
We establish impossibility results showing that analogous
operation-based characterizations do not extend to several other
classes considered in the paper.

\medskip
\noindent\textup{(e)}
We determine the complete inclusion hierarchies among the classes
considered in the paper in the general, Boolean, and two-dimensional
settings.

\medskip

We now describe these five parts in more detail.

\medskip
\noindent\textup{(a): Pairwise closure and TVPI-representability.}
For a family $F$ of operations on $\mathbb Z$, let
$\cl_F$ denote closure under $F$, and define
\[
\Phi_F(S)
=
\Join_{W\in\binom{V}{2}}
\cl_F\bigl(\proj_W(S)\bigr),
\]
where $\proj_W(S)$ denotes the projection of $S$ onto $W$, and
$\Join$ denotes the join operation defined in
\Cref{subsec:Properties-solutions}.
This construction defines a closure operator whose fixed points are
characterized by
\[
S=\Phi_F(S)
\quad\Longleftrightarrow\quad
S\text{ is both $2$-decomposable and $F$-closed}.
\]
We further establish a local-to-global criterion for a prescribed
majority operation $g$: a two-dimensional condition determines
whether, in every dimension, the fixed points of $\Phi_F$ are
precisely the sets closed under $F\cup\{g\}$. This criterion yields,
in a unified way, fixed-point characterizations for sets closed under
a prescribed majority operation and for DC-representable sets.
Combined with the UTVPI characterization in \textup{(b)}, this
framework also identifies the UTVPI-representable sets as precisely
the fixed points of $\Phi_{\avgd}$, where $\avgd$ denotes the directed
midpoint operation introduced in discrete convex
analysis~\citep{TaT21}.\footnote{The term \emph{directed discrete
midpoint operation} is used in \citep{TaT21}; we abbreviate it to
\emph{directed midpoint operation}.}

We also introduce a closed-convex-hull analogue of this pairwise
closure construction and prove that its fixed points are precisely the
TVPI-representable sets. Moreover, if every two-coordinate projection
of $S$ is closed-hole-free, then $2$-decomposability,
TVPI-representability, and closedness under some majority operation
are equivalent. Thus, the pairwise-closure framework also explains how
pairwise convex information governs global TVPI-representability.

\medskip
\noindent\textup{(b): Central characterization of
UTVPI-representability.}
Our central theorem gives the following equivalent characterizations:
\begin{align}
&S\text{ is UTVPI-representable}\label{eq--11}\\
&\quad\Longleftrightarrow\quad
S\text{ is closed under both $\avgd$ and the median operation}\label{eq--22}\\
&\quad\Longleftrightarrow\quad
S\text{ is integrally convex and $2$-decomposable}.\label{eq--33}
\end{align}
The median condition may equivalently be replaced by closedness under
some majority operation. Together with the fixed-point result in
\textup{(a)}, the theorem links five viewpoints on
UTVPI-representability: inequality representation, closedness under
operations, integral convexity, $2$-decomposability, and fixed points
of pairwise closure.

Among these implications,
\eqref{eq--22}~$\Rightarrow$~\eqref{eq--33}
is an immediate consequence of known results, whereas
\eqref{eq--33}~$\Rightarrow$~\eqref{eq--11}
is obtained by combining known results with the lifting argument in the proof of \Cref{thm:main}. 
The principal new ingredient in
\eqref{eq--11}~$\Rightarrow$~\eqref{eq--22}
is the directed-midpoint closedness of UTVPI-representable sets;
their median-closedness had already been established. Consequently,
$2$-decomposability is exactly the additional global condition needed
to extend the two-dimensional characterization to arbitrary dimension.
The complete hierarchy in \textup{(e)} places this characterization
in a broader intersection-theoretic context and reveals a contrasting
behavior for TVPI-representability.

\medskip
\noindent\textup{(c): Multioperation characterizations of
SVPI-representability.}
The SVPI class exhibits a different phenomenon.
We show that no
family of ordinary operations characterizes SVPI-representability by
closedness. We then introduce two natural binary multioperations, one
associated with the integer box spanned by two points and the other
with the integer neighborhood of their midpoint, and show that
closedness under either multioperation characterizes
SVPI-representability.

As is clear from the definition of SVPI systems, SVPI-representable
sets have a Cartesian-product structure. Consequently, independently
of the midpoint-neighborhood characterization, they can be
reconstructed from their one-coordinate projections alone, in
contrast to the pairwise constructions in \textup{(a)}. The box
characterization also follows naturally from the same structure. The
midpoint-neighborhood characterization, by contrast, requires an
additional argument.

\medskip
\noindent\textup{(d): Limits of operation-based
characterizations.}
We next identify the limits of analogous operation-based
characterizations for TVPI-representable, closed-hole-free, hole-free,
$2$-decomposable, and integrally convex sets. The first four classes
are not preserved under coordinate projections. Consequently, none of
them admits a characterization, uniform over all dimensions, by
closedness under a fixed family of total multioperations.

For integrally convex sets, the impossibility result is stronger. For
every $n\geq3$, the class of integrally convex subsets of
$\mathbb Z^n$ is not closed under intersections and therefore cannot
be characterized by closedness under any family of multioperations on
$\mathbb Z$, whether or not they are total. Since ordinary operations
are singleton-valued total multioperations, both impossibility results
apply, in particular, to ordinary operations.

\medskip
\noindent\textup{(e): Complete inclusion hierarchies.}
Finally, we determine all inclusion relations among the
representability, convexity, and operation-closed classes considered
in the paper.
 We give the complete Hasse diagram in the general
setting for $n\geq3$ and determine the corresponding hierarchies over
the Boolean domain and in dimension two, where several classes
coincide. 
For $n\geq3$, we also obtain the complete list of nontrivial pairwise
intersection characterizations of UTVPI-representability among the
classes considered in this paper. In contrast, no analogous pairwise
intersection characterization exists for TVPI-representability.

\subsection*{Related work}

Operation-based characterizations over finite domains and their
connections with decomposability and local consistency are classical;
see, e.g., \citep{McK43,Hor51,Sch78,JCG95,JCC98,CJJ00}. Related
representation--preservation correspondences over infinite domains
include semilinear Horn representations, tropical convexity, and
median-closed semilinear relations~\citep{BoM18a,BoM18b}. In discrete
convex analysis, DC systems are closely related to
$\mathrm{L}^{\natural}$-convexity and rounded-midpoint
operations~\citep{Murota03}; directed discrete midpoint convexity was
introduced in~\citep{TaT21}; and integral convexity and its behavior
under projection have been studied in, e.g.,
\citep{MMTT19,MoM19,MT23}.

UTVPI systems form a well-studied class lying between DC and general
TVPI systems. Their feasibility, structural properties, and
algorithmic applications have been studied extensively; see, among
others,
\citep{Sch91,GuP92,JMS94,Sub04,LaM05,Min06,BKP11,BOS13,UpC13,
SuW17,DMZ18,KS21,KN24}. They are also related to the
Edmonds--Johnson property~\citep{EJ73,GS86,CGZ07,DZ09,DMZ18}.
The present work is complementary to this literature, focusing on
structural characterizations of the integer solution sets of UTVPI
and related systems.

\subsection*{Organization of the paper}

The remainder of this paper is organized as follows.
\Cref{sec:preliminaries} introduces $2$-decomposability, closedness
under operations, midpoint operations, notions of discrete convexity,
and representability by SVPI, DC, UTVPI, and TVPI systems.
\Cref{sec:2-decomp-closure} develops pairwise closure operators,
characterizes their fixed points, and establishes the TVPI fixed-point
characterization and related local-to-global results.
\Cref{sec:polyhedra-closure} 
presents five equivalent
characterizations of UTVPI-representability, describes SVPI-representability in terms of multioperations and a
one-dimensional reconstruction formula, and
determines the limits of analogous operation-based characterizations
for the remaining classes.
\Cref{sec:class-hierarchy} determines the complete inclusion
hierarchies among the classes considered in this paper in the general,
Boolean, and two-dimensional settings. 
In the general setting, the section also identifies all nontrivial
pairs of these classes whose intersection is precisely the UTVPI class
and shows that no analogous pair exists for the TVPI class.
 The separating examples needed to establish
the strict inclusions and incomparabilities in these hierarchies are
collected in the appendix.

\section{Preliminaries}
\label{sec:preliminaries}
In this section, we introduce the basic notions and notation used throughout the paper, including $2$-decomposability, closedness, closure, polymorphisms, convexity, and inequality representability.

\subsection{\texorpdfstring{$2$}{2}-decomposability}
\label{subsec:Properties-solutions}
For a subset $W \subseteq V$ and 
 a vector 
$x = (x_1,\dots,x_n) \in \mathbb{Z}^V$, let 
$\proj_{W}(x)$ denote the projection of $x$ onto $W$, that is,
\[
  \proj_{W}(x)
    = (x_j \mid j \in W) \ \in \ \mathbb{Z}^W.
\]
Here and throughout the paper, we use the natural identification of
$\mathbb Z^W$ with $\mathbb Z^{|W|}$.\footnote{More precisely, whenever
$W=\{i_1,\ldots,i_r\}$ with $i_1<\cdots<i_r$, we identify
$(x_j\mid j\in W)$ with $(x_{i_1},\ldots,x_{i_r})\in\mathbb Z^r$.}
For a set $S \subseteq \mathbb{Z}^V$, we extend this notation by defining
\[
  \proj_{W}(S)
    = \{ \proj_{W}(x) \mid x \in S\} \ 
    \subseteq \ \mathbb{Z}^W.
\]
For each  set $W \in \binom{V}{2}$, let 
$T_W$ be  a subset of  $\mathbb{Z}^W$.
We define the \emph{join} of the family $\{T_W \mid W\in\binom{V}{2}\}$ by

\[
  \Join_{W \in \binom{V}{2}} T_W
    = \Bigl\{ x \in \mathbb{Z}^V 
        \mid \proj_{W}(x) \in T_{W}
        \text{ for all } W \in  \binom{V}{2}\Bigr\}.
\]
It is immediate from the definition that \begin{equation}
\label{eq-111} S \subseteq \Join_{W \in 
\binom{V}{2}}\proj_{W}(S)
\end{equation}
for any $S \subseteq \mathbb{Z}^V$. 
For example, let $V=\{1,2,3\}$ and
$
S=\{(0,0,0),(0,1,1)$, $(1,0,1),(1,1,0)\}
   \subseteq \mathbb{Z}^3.
$
Then
\[
\proj_{\{1,2\}}(S)
=\proj_{\{1,3\}}(S)
=\proj_{\{2,3\}}(S)
=\{0,1\}^2, 
\]
and hence 
$
\Join_{W\in\binom{V}{2}}\proj_W(S)
=
\{0,1\}^3
$. Therefore, we have  
\[
S
\subsetneq
\Join_{W\in\binom{V}{2}}\proj_W(S).
\]
The sets for which equality holds are called $2$-decomposable.
\begin{definition}[$2$-decomposability~\citep{JCC98}]
A set $S \subseteq \mathbb{Z}^V$ is \emph{$2$-decomposable} if 
$S = \Join_{W\in\binom{V}{2}}\proj_W(S)$.
\end{definition}

As discussed in the introduction, $2$-decomposability plays a
fundamental role in the theory of constraint satisfaction problems.

\subsection{Closure systems and closed sets}
\label{subsec-closure}

A family $\mathcal{C}\subseteq 2^{\mathbb{Z}^V}$ is called a
\emph{closure system} if it is closed under arbitrary intersections, that is,
\[
\bigcap_{D \in \mathcal{D}}D\in\mathcal{C}
\]
for every subfamily $\mathcal{D}\subseteq\mathcal{C}$.
The members of $\mathcal{C}$ are called \emph{$\mathcal{C}$-closed sets}.
We next introduce the notion of closedness under operations.

Let $k$ be a positive integer. 
For a $k$-ary operation $f:\mathbb{Z}^k \to \mathbb{Z}$ and vectors $x^{(1)},\dots , x^{(k)} \in \mathbb{Z}^V$, define $f(x^{(1)},\dots , x^{(k)})$ componentwise by 
\[f(x^{(1)},\dots , x^{(k)}) = \bigl(f(x_j^{(1)},\dots, x_j^{(k)}) \mid j \in V \bigr),\] where $x^{(i)} = (x_1^{(i)},\dots , x_n^{(i)})$ for each $i=1,\dots,k$.

\begin{definition}[Closedness under operations]
Let $S$ be a subset of $\mathbb{Z}^n$. 
For an operation $f:\mathbb{Z}^k \to \mathbb{Z}$,  we say that 
 $S$ is  \emph{closed under} $f$ $($or simply \emph{$f$-closed}$)$ 
if  $f(x^{(1)},\dots , x^{(k)}) \in S$  for all $x^{(1)},\dots , x^{(k)} \in S$. 

More generally, 
for a family of operations $F$, 
we say that a set $S\subseteq\mathbb Z^V$ is 
\emph{closed under $F$} $($or simply \emph{$F$-closed}$)$
if it is closed under every operation in $F$.
\end{definition}
By definition, 
the family of all $F$-closed subsets of 
$\mathbb Z^V$ forms a closure system.
Hence, for every set $S \subseteq \mathbb Z^V$, the $F$-closure of $S$  is given  by   
\[\cl_F(S)=\bigcap
\{
T\subseteq\mathbb{Z}^V
\mid
T\supseteq S
\text{ and } T \text{ is } F\text{-closed}
\},\]
that is, $\cl_F(S)$ is the smallest $F$-closed set containing $S$.
It is well known that  $\cl_F$ is indeed a \emph{closure operator}; it is extensive 
(i.e., $S \subseteq \cl_F(S)$), monotone (i.e., $S \subseteq T$ implies 
$\cl_F(S) \subseteq \cl_F(T)$), and idempotent 
(i.e., $\cl_F(\cl_F(S)) = \cl_F(S)$).
For simplicity, we write $\cl_f(S)$ 
instead of $\cl_{\{f\}}(S)$.

Closedness under operations can be viewed from two perspectives:
we may either fix an operation and consider the sets closed under it,
or fix a set and consider the operations under which it is closed.
The latter leads to the following definition.
An operation $f$ is called a \emph{polymorphism} of $S$ if $S$ is $f$-closed.
Majority and median operations are standard examples of polymorphisms that have been extensively studied in the literature.

A ternary operation $f:\mathbb{Z}^3\to\mathbb{Z}$ is called a
\emph{majority} if, for all $x,y\in\mathbb{Z}$,
\[
f(x,x,y)=f(x,y,x)=f(y,x,x)=x.
\]
In other words, a majority operation returns the majority value whenever one exists; otherwise, its value may be chosen arbitrarily. For example, every majority operation satisfies
$
f(1,2,1)=1$,
whereas the value of $f(1,2,3)$  may be any integer.  
The  operation $\median: \mathbb{Z}^3 \to \mathbb{Z}$, which returns the median  of its three arguments, 
is a canonical example of a majority operation.

The following lemma is an immediate consequence of a theorem of Jeavons--Cohen--Cooper relating $2$-decomposability to closedness under operations \citep{JCC98}.\footnote{Although Theorem~3.5 in \citep{JCC98} is stated for finite domains, Section~4.4 in \citep{JCC98} observes that it extends to infinite domains as well.}

\begin{lemma}[\citep{JCC98}]
\label{thm:near-unanimity-closed-is-decompo}
Every set $S \subseteq \mathbb{Z}^n$ that is closed under some majority operation is $2$-decomposable.
\end{lemma}
In general, the converse of
\Cref{thm:near-unanimity-closed-is-decompo} does not hold.
In fact, a stronger non-characterizability result for the class of
$2$-decomposable sets is proved in
\Cref{thm:limits-operation-characterizations}~\textup{(i)}.

We next introduce several midpoint-related operations.
Note that, for two integers $x$ and $y$, the midpoint
$
\frac{x+y}{2}
$
need not be integral. This gives rise to several natural integral midpoint operations.
Binary operations $m^+$ and $m^-$ 
are called the
\emph{upper midpoint} and the
\emph{lower midpoint}, respectively, and are defined by
\[
m^+(x,y)=\bigl\lceil \tfrac{x+y}{2} \bigr\rceil \mbox{ and }m^-(x,y)=\bigl\lfloor \tfrac{x+y}{2} \bigr\rfloor
\]for any integers $x$ and $y$. 
Another example is the
\emph{directed midpoint} operation~\citep{TaT21}, defined for
any integers $x$ and $y$ by
\[
\avgd(x,y)=
\begin{cases}
\left\lceil\dfrac{x+y}{2}\right\rceil,
& \text{if }x\ge y,\\[8pt]
\left\lfloor\dfrac{x+y}{2}\right\rfloor,
& \text{if }x<y.
\end{cases}
\]
In other words, it rounds $(x+y)/2$ toward $x$ whenever $(x+y)/2$ is not integral.
These midpoint operations are related through their induced closure
properties. In particular, closedness under both $m^+$ and $m^-$
implies $\avgd$-closedness~\citep[Corollary~2]{TaT21}.

\subsection{Convexity and inequality representations}
\label{subsec:Polyhedral-representation}
In this subsection, we next introduce the notions of hole-free and integrally convex sets, followed by convex sets represented by systems of linear inequalities.

For a set $S\subseteq\mathbb{R}^n$, let $\conv(S)$ denote the
\emph{convex hull} of $S$, that is, the intersection of all convex sets
containing $S$. We denote by $\overline{\conv}(S)$ the
\emph{closed convex hull} of $S$, that is, the intersection of all
(topologically) closed convex sets containing $S$.
Both $\conv$ and $\overline{\conv}$ are closure operators on $2^{\mathbb{R}^n}$. 
A set $S\subseteq\mathbb{Z}^V$ is  called \emph{hole-free} \citep{Murota03}  if
$
S=\conv(S)\cap\mathbb{Z}^V $, 
and \emph{closed-hole-free} if
$
S=\overline{\conv}(S)\cap\mathbb{Z}^V $. 
Equivalently, $S$ coincides with the set of  integer points in its convex hull and closed convex hull, respectively.
Since every $S \subseteq  \mathbb{Z}^V$ satisfies $S \subseteq \conv(S)\cap\mathbb{Z}^V \subseteq \overline{\conv}(S)\cap\mathbb{Z}^V$ by definition, every closed-hole-free set is  hole-free. 
Although the converse fails in general, it holds whenever $S$ is finite.

For $x \in \mathbb{R}^V$, the \emph{integer neighborhood} of $x$ is
\[
  N(x) = \{ z \in \mathbb{Z}^V \mid |z_j - x_j| < 1 \text{ for } j \in V \}.
\]
A set $S \subseteq \mathbb{Z}^V$ is called \emph{integrally convex} if 
 every $x\in\conv(S)$
satisfies  $x \in \conv\bigl(S\cap N(x)\bigr)$.
By definition, every integrally convex set is hole-free, since
$N(x)=\{x\}$ for every $x\in\mathbb Z^V$.
Moreover, it is known that the convex hull of every integrally convex set
is closed~\citep{MT20}. Hence, every integrally convex set is closed-hole-free.
Consequently, the corresponding classes of sets satisfy
\begin{equation}
\label{eq-convexrelations}
\text{integrally convex} \ \, 
\subsetneq \, \ 
\text{closed-hole-free}
\ \, \subsetneq \ \, 
\text{hole-free}.
\end{equation}

The following theorem shows that it suffices to consider only midpoints in the definition of integral convexity.
\begin{lemma}[Theorem 2.1 in~\citep{MT23}]
\label{lemma:IntConvSet-characterization}
A set $S \subseteq \mathbb{Z}^V$ is integrally convex if and only if 
\begin{align}
\frac{y+z}{2} \in \conv\left(S\cap N\left(\frac{y+z}{2}\right)\right)
\end{align}
for every $y,z \in S$.
\end{lemma}
We next define convex sets defined by linear systems. 
Let
\begin{eqnarray}
\mathcal{L}&=&\{a^{(i)}x\le b^{(i)}\mid i\in I\}
\label{eq-linearsystem}
\end{eqnarray}
be a \emph{linear system} (i.e., a set of linear inequalities), 
indexed by a possibly infinite set $I$, where $a^{(i)}\in\mathbb{R}^V$ and $b^{(i)}\in\mathbb{R}$ for every $i\in I$.
For $a\in\mathbb{R}^V$, let
\begin{eqnarray*}
\supp(a)&=&\{j\in V\mid a_j\neq 0\}.
\end{eqnarray*}

\begin{definition}
Let $\mathcal{L}$ be a linear system given in $($\ref{eq-linearsystem}$)$. Then  it is called 
\begin{description}
\setlength{\itemsep}{-1pt}
\setlength{\parsep}{0pt}
\item[{\rm TVPI}] if $|\supp(a^{(i)})|\leq 2$ for every $i\in I$;

\item[{\rm  UTVPI}] if it is TVPI and   
$a^{(i)}\in \{0,-1,+1\}^V$ for every $i\in I$;

\item[{\rm DC}] if it is UTVPI and  
 every $a^{(i)}$ has at most one
component equal to $+1$, at most one component equal to $-1$;

\item[{\rm SVPI}] if $|\supp(a^{(i)})|\leq 1$ for every $i\in I$.
\end{description}
Here {\rm TVPI}, {\rm UTVPI}, {\rm DC}, and {\rm SVPI} stand for
\emph{two-variable-per-inequality},
\emph{unit two-variable-per-inequality},
\emph{difference constraint}, and
\emph{single-variable-per-inequality}, respectively.
\end{definition}
By convention, every SVPI system $\mathcal{L}$ is assumed to satisfy
that 
$a^{(i)}\in \{0,-1,+1\}^V$ for every $i\in I$. 
Accordingly, we regard every SVPI system as a UTVPI system.
Hereafter, we use TVPI, UTVPI, DC, and SVPI to denote the corresponding classes of linear systems. Then 
\[
\mathrm{SVPI}
\subsetneq
\mathrm{DC}
\subsetneq
\mathrm{UTVPI}
\subsetneq
\mathrm{TVPI}.
\]
We say that  $S \subseteq \mathbb{Z}^V$ is
linearly representable 
(resp., \emph{TVPI-representable}, \emph{UTVPI-representable},
\emph{DC-representable},
\emph{SVPI-representable})
if there exists a  linear (resp., TVPI, UTVPI, DC, SVPI) system $\mathcal{L}$ such that 
\begin{eqnarray*}
S&=&\{ x\in \mathbb{Z}^V \mid x \mbox{ satisfies }  \mathcal{L}\}. 
\end{eqnarray*}
Such a system $\mathcal{L}$ is called 
a \emph{linear} (resp., TVPI, UTVPI, DC, SVPI) \emph{representation} of $S$.

Although a UTVPI system $\mathcal L$ may contain infinitely many
inequalities, it is always equivalent to a finite UTVPI system.
Indeed, for each coefficient vector $a$, let $B_a$ be the set of
right-hand sides of the inequalities in $\mathcal L$ having
coefficient vector $a$. If $B_a\neq\emptyset$ and
$\inf B_a>-\infty$, then all such inequalities can be replaced by
the single inequality
\[
a x\leq\inf B_a.
\]
If $\inf B_a=-\infty$ for some $a$, then $\mathcal L$ is infeasible.
Since there are only $O(n^2)$ possible UTVPI coefficient vectors,
this yields an equivalent finite UTVPI system.

Consequently, every UTVPI system defines a possibly empty polyhedron
\[
P
=
\left\{
x\in\mathbb R^V
\ \middle|\
x\text{ satisfies }\mathcal L
\right\}.
\]
Thus, every UTVPI-representable set is the set of integer points of a
UTVPI polyhedron, that is,
\[
S=P\cap\mathbb Z^V.
\]

The same argument also applies to SVPI and DC systems.
Thus, every SVPI, DC, and UTVPI system is equivalent to a finite system of the same type, and hence defines a polyhedron.
In contrast, this is no longer true for TVPI systems.
Allowing infinitely many inequalities is essential only for general linear and TVPI systems.

A classical result in convex analysis gives the following geometric characterization of linear representability. 
\begin{proposition}[Theorem 11.5 in~\citep{Roc97}]
\label{prop:linear-holefree}
A set $S\subseteq\mathbb Z^V$ is linearly representable
 if and only if it is closed-hole-free. 
\end{proposition}

The following proposition shows that, in dimension \emph{two}, integral convexity, which is stronger than closed-hole-freeness, also admits a complete characterization in terms of systems of linear inequalities.

\begin{proposition}[{\citep[Proposition 2.1]{MMTT19}}]
\label{prop:two-dim-IntConv=UTVPI}
A set $S \subseteq \mathbb{Z}^2$ is  UTVPI-representable if and only if it is integrally convex.
\end{proposition}
The proposition above shows that, in dimension two, integral convexity
exactly characterizes UTVPI-representability. In higher dimensions,
however, integral convexity alone is no longer sufficient. Our main
theorem identifies $2$-decomposability as the precise additional global
condition: a set is UTVPI-representable if and only if it is both
integrally convex and $2$-decomposable; see \Cref{thm:main}.

The following two results connect inequality representability with the
operation-based viewpoint introduced in \Cref{subsec-closure}.

\begin{theorem}[{\citep[Section~5.5]{Murota03}}]
\label{thm:DC-midpoint-characterization}
A set $S\subseteq\mathbb Z^V$ is DC-representable if and only if it is
closed under both $m^+$ and $m^-$.
\end{theorem}

Sets satisfying these equivalent conditions are known as $\mathrm{L}^{\natural}$-convex sets in discrete convex analysis.
Since our focus is on closedness under operations, we omit the standard definition of $\mathrm{L}^{\natural}$-convexity; see \citep{Murota03} for further details.

\begin{proposition}[{\citep[Theorem~3.2]{BoM18b}}]
\label{prop:TVPI-is-median-closed}
Every TVPI-representable set is median-closed.%
\footnote{Although \citep[Theorem~3.2]{BoM18b} is stated for finite
systems, the same proof applies to possibly infinite systems.}
\end{proposition}

\section{Pairwise closure and local-to-global structure}
\label{sec:2-decomp-closure}
In this section, we develop pairwise closure constructions and study
their interaction with $2$-decomposability. We first construct closure
operators from the closures of pairwise projections. We then characterize
their fixed points and establish local-to-global principles for closure
under operations. Finally, we specialize the framework to the closure
induced by closed convex hulls.


\subsection{Pairwise closure operators}
\label{subsec:pairwise-closure}
We begin with an elementary identity. 
For a $k$-ary operation
$f:\mathbb Z^k\to\mathbb Z$, a subset $W\subseteq V$, and vectors
$x^{(1)},\ldots,x^{(k)}\in\mathbb Z^V$, componentwise application of
$f$ commutes with projection:

\begin{equation}
\label{eq-proj}
\proj_W\bigl(f(x^{(1)},\ldots,x^{(k)})\bigr)
=
f\bigl(
\proj_W(x^{(1)}),\ldots,\proj_W(x^{(k)})
\bigr).
\end{equation}
Together with the definition of the join, this identity yields the
following preservation properties.

\begin{lemma}
\label{lemma:proj-closed}
Let $F$ be a family of operations on $\mathbb Z$.
Then the following statements hold.
\begin{enumerate}
\setlength{\itemsep}{4pt}
\setlength{\parsep}{0pt}
\item[\textup{(i)}]
If a set $S\subseteq\mathbb Z^V$ is $F$-closed, then
$\proj_W(S)$ is $F$-closed for every $W\subseteq V$.

\item[\textup{(ii)}]
If $T_W\subseteq\mathbb Z^W$ is $F$-closed for every
$W\in\binom{V}{2}$, then
$
\Join_{W\in\binom{V}{2}}T_W
$
is $F$-closed.

\item[\textup{(iii)}] If $S\subseteq\mathbb Z^V$ is $F$-closed, then
$
\Join_{W\in\binom{V}{2}}\proj_W(S)
$
is also $F$-closed.
\end{enumerate}
\end{lemma}

\begin{proof}
It suffices to prove each statement for an arbitrary operation $f\in F$,
since a set is $F$-closed if and only if it is $f$-closed for every
$f\in F$. Fix an arbitrary $k$-ary operation $f\in F$.

For \textup{(i)}, let
$
\proj_W(x^{(1)}),\ldots,\proj_W(x^{(k)})
$
be arbitrary vectors in $ \proj_W(S)$, with $x^{(1)},\ldots,x^{(k)}\in S$.
Since $S$ is $f$-closed, we have
$
f(x^{(1)},\ldots,x^{(k)})\in S.
$
Hence, by~\eqref{eq-proj},
\[
\begin{aligned}
f\bigl(
\proj_W(x^{(1)}),\ldots,\proj_W(x^{(k)})
\bigr)
&=
\proj_W\bigl(f(x^{(1)},\ldots,x^{(k)})\bigr)\
&\in \  \proj_W(S),
\end{aligned}
\]
implying that $\proj_W(S)$ is $f$-closed.

For \textup{(ii)}, let
$
x^{(1)},\ldots,x^{(k)}
$ be vectors in 
$\Join_{W\in\binom{V}{2}}T_W$. 
For every $W\in\binom{V}{2}$ and every $\ell=1,\ldots,k$, we have
$\proj_W(x^{(\ell)})\in T_W$. Since $T_W$ is $f$-closed, \eqref{eq-proj} gives
\[
\proj_W\bigl(f(x^{(1)},\ldots,x^{(k)})\bigr) \ = \  f\bigl(
\proj_W(x^{(1)}),\ldots,\proj_W(x^{(k)})
\bigr)
\ \in \ T_W 
\]
for every $W\in\binom{V}{2}$, which implies
\[
f(x^{(1)},\ldots,x^{(k)})
\in
\Join_{W\in\binom{V}{2}}T_W.
\]
Hence, the join is $f$-closed.

Finally, \textup{(iii)} follows from \textup{(i)} and \textup{(ii)}
by setting
$T_W=\proj_W(S)$ for all $W\in\binom{V}{2}$.
\end{proof}

We now use these preservation properties to construct a closure operator
from the closures of pairwise projections.

Let $F$ be a family of operations on $\mathbb Z$. Define
$
\Phi_F:2^{\mathbb Z^V}\to 2^{\mathbb Z^V}
$
by
\[
\Phi_F(S)
=
\Join_{W\in\binom{V}{2}}
\cl_F\bigl(\proj_W(S)\bigr).
\]
When $F=\{f\}$, we simply write $\Phi_f$ for $\Phi_{\{f\}}$.

\begin{lemma}
\label{lemma:2-decompo-closure-operator}
For a family $F$ of operations on $\mathbb Z$, $\Phi_F$ is a closure operator on $2^{\mathbb Z^V}$.
\end{lemma}

\begin{proof}
We verify that $\Phi_F$ is extensive, monotone, and idempotent.

\emph{Extensivity.}
For every $W\in\binom{V}{2}$, we have
$\proj_W(S)\subseteq\cl_F\bigl(\proj_W(S)\bigr)$.
It therefore follows from the definition of the join that
$S\subseteq\Phi_F(S)$.

\emph{Monotonicity.}
Let $S_1, S_2$  be subsets  of ${\mathbb Z}^V$ such that  $S_1\subseteq S_2$. Then 
$
\proj_W(S_1)\subseteq\proj_W(S_2)
$
for every $W\in\binom{V}{2}$. By the monotonicity of $\cl_F$,
\[
\cl_F\bigl(\proj_W(S_1)\bigr)
\subseteq
\cl_F\bigl(\proj_W(S_2)\bigr).
\]
By taking the joins over all $W\in\binom{V}{2}$ we obtain 
$
\Phi_F(S_1)\subseteq\Phi_F(S_2).
$

\emph{Idempotence.}
By extensivity and monotonicity, we have 
$
\Phi_F(S)\subseteq\Phi_F\bigl(\Phi_F(S)\bigr).
$
For the reverse inclusion, observe that
$
\proj_W\bigl(\Phi_F(S)\bigr)
\subseteq
\cl_F\bigl(\proj_W(S)\bigr)
$ for every
$W\in\binom{V}{2}$. 
Therefore,
\[
\begin{aligned}
\cl_F\bigl(\proj_W(\Phi_F(S))\bigr)
\ \subseteq \ 
\cl_F\bigl(\cl_F(\proj_W(S))\bigr)
\ =\ 
\cl_F\bigl(\proj_W(S)\bigr).
\end{aligned}
\]
Taking the joins over all $W\in\binom{V}{2}$ yields
$
\Phi_F\bigl(\Phi_F(S)\bigr)
\subseteq
\Phi_F(S).
$
Thus, $\Phi_F$ is idempotent.
\end{proof}

\subsection{Fixed points and local-to-global characterizations}
\label{subsec:operation-induced-closure}

We first characterize the fixed points of $\Phi_F$ in terms of
$2$-decomposability and $F$-closedness. We then establish a
local-to-global criterion characterizing when these fixed points are
precisely the sets closed under $F$ and a prescribed majority operation.

\begin{proposition}[Fixed points of $\Phi_F$]
\label{prop:fixed-points-Phi-f}
Let $F$ be a family of operations on $\mathbb Z$, and let
$S\subseteq\mathbb Z^V$. Then the following statements are equivalent:
\begin{enumerate}
\setlength{\itemsep}{4pt}
\setlength{\parsep}{0pt}

\item[\textup{(i)}]
$S$ is a fixed point of $\Phi_F$, that is,
$
S=\Phi_F(S).
$

\item[\textup{(ii)}]
$S$ is both $2$-decomposable and $F$-closed.
\end{enumerate}
\end{proposition}

\begin{proof}
Suppose first that $S=\Phi_F(S)$. Since
\[
S
\subseteq
\Join_{W\in\binom{V}{2}}\proj_W(S)
\subseteq
\Phi_F(S)
=
S,
\]
we have
\[
S
=
\Join_{W\in\binom{V}{2}}\proj_W(S),
\]
and hence $S$ is $2$-decomposable. Moreover, since each set
$\cl_F(\proj_W(S))$ is $F$-closed, 
\Cref{lemma:proj-closed}\textup{(ii)} implies that
$\Phi_F(S)$, and hence $S$, is $F$-closed.

Conversely, suppose that $S$ is both $2$-decomposable and $F$-closed.
By \Cref{lemma:proj-closed}\textup{(i)}, every
$\proj_W(S)$ is $F$-closed. Hence,
$\cl_F\bigl(\proj_W(S)\bigr)=\proj_W(S)$ 
for every $W\in\binom{V}{2}.
$
Therefore,
\[
\Phi_F(S)
=
\Join_{W\in\binom{V}{2}}\proj_W(S)
=
S.
\]
\end{proof}

\begin{corollary}
\label{thm:2-decompo-closure-f}
Let $F$ be a family of operations on $\mathbb Z$, and let
$S\subseteq\mathbb Z^V$ be $2$-decomposable.
Then the following statements are equivalent:
\begin{enumerate}
\setlength{\itemsep}{4pt}
\setlength{\parsep}{0pt}

\item[\textup{(i)}]
$S$ is $F$-closed.

\item[\textup{(ii)}]
$\proj_W(S)$ is $F$-closed for every $W\in\binom{V}{2}$.

\item[\textup{(iii)}]
$S=\Phi_F(S)$.
\end{enumerate}
\end{corollary}

\begin{proof}
The implication \textup{(i)} $\Rightarrow$ \textup{(ii)} follows from
\Cref{lemma:proj-closed}\textup{(i)}.

Conversely, suppose that \textup{(ii)} holds. Since $S$ is
$2$-decomposable,
\[
S
=
\Join_{W\in\binom{V}{2}}\proj_W(S).
\]
Therefore, \Cref{lemma:proj-closed}\textup{(ii)} implies that
$S$ is $F$-closed, proving
\textup{(ii)} $\Rightarrow$ \textup{(i)}.

Finally, the equivalence of \textup{(i)} and \textup{(iii)}
follows from \Cref{prop:fixed-points-Phi-f}, since $S$ is
$2$-decomposable.
\end{proof}

Since every set closed under a majority operation is
$2$-decomposable by
\Cref{thm:near-unanimity-closed-is-decompo},
\Cref{thm:2-decompo-closure-f} applies, in particular, to every such set.

We next characterize when the fixed points of $\Phi_F$ are precisely
the sets that are $(F\cup\{g\})$-closed for a prescribed majority
operation $g$.

\begin{theorem}
\label{thm:2-decompo-closure-special-g}
Let $F$ be a family of operations on $\mathbb Z$, and let $g$ be a
majority operation. Then the following statements are equivalent:
\begin{enumerate}
\setlength{\itemsep}{4pt}
\setlength{\parsep}{0pt}

\item[\textup{(i)}]
Every $F$-closed subset of $\mathbb Z^2$ is $g$-closed.

\item[\textup{(ii)}]
For every set $S\subseteq\mathbb Z^V$,
$
S=\Phi_F(S)
$
if and only if $S$ is $(F\cup\{g\})$-closed.
\end{enumerate}
\end{theorem}

\begin{proof}
Suppose first that \textup{(i)} holds, and let
$S\subseteq\mathbb Z^V$.

Assume that $S=\Phi_F(S)$. For every
$W\in\binom{V}{2}$, let
\[
T_W=\cl_F\bigl(\proj_W(S)\bigr).
\]
The set $T_W$ is $F$-closed by definition. Since $|W|=2$,
\textup{(i)} implies that $T_W$ is $g$-closed. 
Thus, $T_W$ is
$(F\cup\{g\})$-closed.

Applying
\Cref{lemma:proj-closed}\textup{(ii)} to the family of operations
$F\cup\{g\}$, we conclude that
\[
S
=
\Phi_F(S)
=
\Join_{W\in\binom{V}{2}}T_W
\]
is $(F\cup\{g\})$-closed.

Conversely, assume that $S$ is $(F\cup\{g\})$-closed.
Then $S$ is $F$-closed and $g$-closed. Since $g$ is a majority
operation, \Cref{thm:near-unanimity-closed-is-decompo} implies that
$S$ is $2$-decomposable. Therefore,
\Cref{prop:fixed-points-Phi-f} yields
$
S=\Phi_F(S),  
$
which proves \textup{(ii)}.

Now suppose that \textup{(ii)} holds.
 Let
$R\subseteq\mathbb Z^2$ be an arbitrary $F$-closed set, and define
\[
S
=
\{\,x\in\mathbb Z^V
\mid
\proj_{\{1,2\}}(x)\in R\,\}.
\]

Then $\proj_{\{1,2\}}(S)=R$. Moreover, $S$ is $F$-closed, since
membership in $S$ is determined solely by the $F$-closed set $R$ on
the first two coordinates.
By \eqref{eq-111}, 
we have
\[
S
\subseteq
\Join_{W\in\binom{V}{2}}\proj_W(S).
\]
For the reverse inclusion, if
$
x\in
\Join_{W\in\binom{V}{2}}\proj_W(S),
$
then
\[
\proj_{\{1,2\}}(x)
\in
\proj_{\{1,2\}}(S)
=
R,
\]
and hence $x\in S$. Therefore, $S$ is $2$-decomposable. By
\Cref{prop:fixed-points-Phi-f},
\[
S=\Phi_F(S).
\]
It follows from \textup{(ii)} that $S$ is
$(F\cup\{g\})$-closed, and in particular $g$-closed. Hence,
\Cref{lemma:proj-closed}\textup{(i)} implies that
\[
R
=
\proj_{\{1,2\}}(S)
\]
is $g$-closed. Since $R$ is an arbitrary $F$-closed set, we have \textup{(i)}, completing the proof. 
\end{proof}

\begin{lemma}
\label{lem:two-dimensional-midpoint-median}
Every $\avgd$-closed subset of $\mathbb Z^2$ is median-closed.
\end{lemma}

\begin{proof}
Every $\avgd$-closed subset of $\mathbb Z^2$ is integrally convex
by \citep{TaT21}. Hence, it is UTVPI-representable by
\Cref{prop:two-dim-IntConv=UTVPI}, and therefore median-closed by
\Cref{prop:TVPI-is-median-closed}.
\end{proof}

\begin{corollary}[Concrete fixed-point characterizations] \label{cor:concrete-fixed-point-characterizations} For a set $S\subseteq\mathbb Z^V$, the following statements hold: \begin{enumerate} \setlength{\itemsep}{4pt} \setlength{\parsep}{0pt} \item[\textup{(i)}] For every majority operation $g$, the set $S$ is $g$-closed if and only if \[ S=\Phi_g(S). \] 
 In particular, $S$ is median-closed if and only if $S=\Phi_{\median}(S)$. 
\item[\textup{(ii)}] $S$ is closed under both $\avgd$ and the median operation if and only if \[ S=\Phi_{\avgd}(S). \] \item[\textup{(iii)}] $S$ is DC-representable if and only if \[ S=\Phi_{\{m^+,m^-\}}(S). \] \end{enumerate} 
\end{corollary}

\begin{proof}
For \textup{(i)}, fix a majority operation $g$ and apply
\Cref{thm:2-decompo-closure-special-g} with $F=\{g\}$.
The hypothesis of the theorem is immediate, since every
$g$-closed subset of $\mathbb Z^2$ is $g$-closed. Hence,
$
S=\Phi_{g}(S)
$
if and only if $S$ is $g$-closed, which proves the first statement.
The particular case $g=\median$ gives the median statement.

For \textup{(ii)}, apply
\Cref{thm:2-decompo-closure-special-g} with
\[
F=\{\avgd\}
\qquad\text{and}\qquad
g=\median.
\]
By \Cref{lem:two-dimensional-midpoint-median}, every
$\avgd$-closed subset of $\mathbb Z^2$ is median-closed. Therefore,
\[
S=\Phi_{\avgd}(S)
\]
if and only if $S$ is
$\{\avgd,\median\}$-closed, that is, closed under both $\avgd$ and
the median operation.

For \textup{(iii)}, apply the same theorem with
\[
F=\{m^+,m^-\}
\qquad\text{and}\qquad
g=\median.
\]
Since $m^+$- and $m^-$-closedness imply
median-closedness by \Cref{thm:DC-midpoint-characterization} and \Cref{prop:TVPI-is-median-closed},
\[
S=\Phi_{\{m^+,m^-\}}(S)
\]
if and only if $S$ is
$\{m^+,m^-,\median\}$-closed. 
Again by~\Cref{thm:DC-midpoint-characterization} and \Cref{prop:TVPI-is-median-closed}, 
the latter is equivalent to $m^+$- and $m^-$-closedness.
By \Cref{thm:DC-midpoint-characterization}, this is equivalent to DC-representability of $S$, which proves \textup{(iii)}.
\end{proof}

In the next section, we show that the equivalent conditions in
\textup{(ii)} characterize UTVPI-representability. Thus, the fixed
points of $\Phi_{\avgd}$ are precisely the UTVPI-representable sets.

\subsection{Closed-convex-hull closure}
\label{subsec:convex-closure-hierarchy}

We finally specialize the preceding pairwise closure framework to the
closure induced by closed convex hulls. For a set
$S\subseteq\mathbb Z^V$, define
\[
\Psi(S)
=
\Join_{W\in\binom{V}{2}}
\left(
\overline{\conv}\bigl(\proj_W(S)\bigr)\cap\mathbb Z^W
\right).
\]
For every $W\in\binom{V}{2}$, the map
\[
T
\longmapsto
\overline{\conv}(T)\cap\mathbb Z^W
\]
is a closure operator on $2^{\mathbb Z^W}$. Hence, the same argument
as in \Cref{lemma:2-decompo-closure-operator} shows that $\Psi$ is a
closure operator on $2^{\mathbb Z^V}$.

We first state a consequence of
\Cref{prop:linear-holefree} that will be used below.

\begin{lemma}
\label{lem:join-closed-hole-free-TVPI}
Let $R_W\subseteq\mathbb Z^W$ be closed-hole-free for every
$W\in\binom{V}{2}$. Then the join
\[
\Join_{W\in\binom{V}{2}}R_W
\]
is TVPI-representable. Consequently, it is median-closed and
closed-hole-free.
\end{lemma}

\begin{proof}
By \Cref{prop:linear-holefree}, each $R_W$ is linearly representable.
Since $|W|=2$, every linear inequality in a representation of $R_W$
involves at most two variables, and hence $R_W$ is
TVPI-representable.

For every $W\in\binom{V}{2}$, lift a TVPI representation of $R_W$ to
the coordinates in $W$. The union of the resulting systems is a TVPI
representation of
\[
\Join_{W\in\binom{V}{2}}R_W.
\]
Therefore, the join is median-closed by
\Cref{prop:TVPI-is-median-closed} and closed-hole-free by
\Cref{prop:linear-holefree}.
\end{proof}

The fixed points of $\Psi$ admit the following characterization.

\begin{theorem}[TVPI fixed-point characterization]
\label{thm:pairwise-closed-convex-hull}
For a set $S\subseteq\mathbb Z^V$, the following statements are
equivalent:
\begin{enumerate}
\setlength{\itemsep}{4pt}
\setlength{\parsep}{0pt}

\item[\textup{(i)}]
$S$ is TVPI-representable.

\item[\textup{(ii)}]
$S=\Psi(S)$.
\end{enumerate}
\end{theorem}

\begin{proof}
Suppose first that \textup{(ii)} holds. For every
$W\in\binom{V}{2}$, let
\[
T_W
=
\overline{\conv}\bigl(\proj_W(S)\bigr)\cap\mathbb Z^W.
\]
Since the map
\[
T
\longmapsto
\overline{\conv}(T)\cap\mathbb Z^W
\]
is a closure operator, each $T_W$ is closed-hole-free. Moreover,
\[
S
=
\Psi(S)
=
\Join_{W\in\binom{V}{2}}T_W.
\]
Hence, \Cref{lem:join-closed-hole-free-TVPI} implies that $S$ is
TVPI-representable.

Conversely, suppose that $S$ is TVPI-representable. Fix a TVPI
representation of $S$, and assign each inequality to a set
$W\in\binom{V}{2}$ containing the support of that inequality.
For every $W\in\binom{V}{2}$, let
$R_W\subseteq\mathbb Z^W$ be the set of integer solutions of the
inequalities assigned to $W$, where $R_W=\mathbb Z^W$ if no
inequality is assigned to $W$. Then
\[
S
=
\Join_{W\in\binom{V}{2}}R_W.
\]
Each $R_W$ is linearly representable and hence closed-hole-free by
\Cref{prop:linear-holefree}. Since
$\proj_W(S)\subseteq R_W$, we have
\[
\overline{\conv}\bigl(\proj_W(S)\bigr)\cap\mathbb Z^W
\subseteq
\overline{\conv}(R_W)\cap\mathbb Z^W
=
R_W.
\]
It follows that
\[
\Psi(S)
\subseteq
\Join_{W\in\binom{V}{2}}R_W
=
S.
\]
The reverse inclusion $S\subseteq\Psi(S)$ follows from the definition
of $\Psi$. Therefore, $S=\Psi(S)$.
\end{proof}

Combining the preceding theorem with
\Cref{prop:TVPI-is-median-closed,thm:near-unanimity-closed-is-decompo},
we obtain the following local-to-global characterization.

\begin{corollary}
\label{cor:pairwise-closed-hole-free-TVPI}
Suppose that $\proj_W(S)$ is closed-hole-free for every
$W\in\binom{V}{2}$. Then the following statements are equivalent:
\begin{enumerate}
\setlength{\itemsep}{4pt}
\setlength{\parsep}{0pt}

\item[\textup{(i)}]
$S$ is $2$-decomposable.

\item[\textup{(ii)}]
$S$ is TVPI-representable.

\item[\textup{(iii)}]
$S$ is closed under some majority operation.
\end{enumerate}
\end{corollary}

\begin{proof}
Under the assumption on the pairwise projections,
$2$-decomposability implies $S=\Psi(S)$. Hence,
\textup{(i)} $\Rightarrow$ \textup{(ii)} follows from
\Cref{thm:pairwise-closed-convex-hull}.
The implication \textup{(ii)} $\Rightarrow$ \textup{(iii)} follows
from \Cref{prop:TVPI-is-median-closed}, since the median is a majority
operation, and \textup{(iii)} $\Rightarrow$ \textup{(i)} follows from
\Cref{thm:near-unanimity-closed-is-decompo}.
\end{proof}

The assumption on the pairwise projections in the preceding
corollary is not implied by closed-hole-freeness of $S$.

\begin{remark}
\label{rem:closed-hole-free-not-preserved-by-projection}
As the following example shows, closed-hole-freeness is not preserved
under projection. Let
\[
S
=
\{\,(0,0,0),(2,2,3)\,\}
\subseteq
\mathbb Z^3.
\]
The line segment joining the two points contains no other integer
point, and hence $S$ is closed-hole-free. However,
\[
\proj_{\{1,2\}}(S)
=
\{\,(0,0),(2,2)\,\}
\]
is not closed-hole-free, since its closed convex hull contains the
integer point $(1,1)$.

This example will also be used in
\Cref{subsec:limits-operation-characterizations}.
\end{remark}

\section{Operation-based characterizations and their limits}
\label{sec:polyhedra-closure}

In this section, we present operation-based characterizations of the
classes introduced above and determine the extent to which such
characterizations are possible.
  We begin by combining the known
operation-based characterization of DC-representability with the
fixed-point characterization established in
\Cref{sec:2-decomp-closure}. We then establish five equivalent  characterizations of UTVPI-representability. For
completeness, the corresponding fixed-point characterization, already
proved in \Cref{sec:2-decomp-closure}, is included in the theorem so
that all equivalent descriptions are presented together. Next, we
show that, although SVPI-representability cannot be characterized by
ordinary operations, it admits equivalent characterizations in terms
of multioperations and a one-dimensional reconstruction formula.
Finally, we establish limits on operation-based characterizations for
the remaining five classes considered in this paper:
TVPI-representable, closed-hole-free, hole-free, $2$-decomposable,
and integrally convex sets.

The following corollary combines the known operation-based
characterization of DC-representability with the fixed-point results
of \Cref{sec:2-decomp-closure}.

\begin{corollary}[Characterizations of DC-representability]
\label{cor:DC-fixed-point-characterization}
For a set $S\subseteq\mathbb Z^V$, the following statements are
equivalent:
\begin{enumerate}
\setlength{\itemsep}{4pt}
\setlength{\parsep}{0pt}

\item[\textup{(i)}]
$S$ is DC-representable.

\item[\textup{(ii)}]
$S$ is closed under both $m^+$ and $m^-$.

\item[\textup{(iii)}]
$
S=\Phi_{\{m^+,m^-\}}(S).
$
\end{enumerate}
\end{corollary}

\begin{proof}
This follows from
\Cref{thm:DC-midpoint-characterization,cor:concrete-fixed-point-characterizations} (iii).
\end{proof}

\subsection{Five equivalent characterizations of UTVPI-representability}
\label{subsec:main-results}

We now prove the central theorem of the paper. It characterizes
UTVPI-representability in terms of preservation under coordinatewise
operations, integral convexity together with $2$-decomposability, and
a fixed-point property of the pairwise closure operator introduced in
\Cref{sec:2-decomp-closure}. Thus, the theorem brings together the five
viewpoints highlighted in the Introduction.

\begin{theorem}[Five equivalent characterizations of
UTVPI-representability]
\label{thm:main}
For a set $S\subseteq\mathbb Z^V$, the following statements are
equivalent:
\begin{enumerate}
\setlength{\itemsep}{4pt}
\setlength{\parsep}{0pt}

\item[\textup{(i)}]
$S$ is UTVPI-representable.

\item[\textup{(ii)}]
$S$ is closed under both $\avgd$ and the median operation.

\item[\textup{(iii)}]
$S$ is $\avgd$-closed and closed under some majority operation.

\item[\textup{(iv)}]
$S$ is integrally convex and $2$-decomposable.

\item[\textup{(v)}]
$S=\Phi_{\avgd}(S)$.
\end{enumerate}
\end{theorem}

The proof uses three auxiliary results. We first recall that
directed-midpoint closedness implies integral convexity. 
We then establish an elementary inequality that yields
directed-midpoint closedness for UTVPI-representable sets. Together
with their known median-closedness, this proves \textup{(ii)}.
 The remaining
implications follow by combining known results on majority operations,
integral convexity under projection, and two-dimensional
UTVPI-representability with a projection-and-lifting argument and the
fixed-point characterization established in
\Cref{sec:2-decomp-closure}.

The following implication is proved in a more general functional
setting in \citep{TaT21}; we include a short proof for completeness.

\begin{lemma}
\label{lem:avgd=>int-conv}
If a set $S\subseteq\mathbb Z^V$ is $\avgd$-closed, then it is
integrally convex.
\end{lemma}

\begin{proof}
By \Cref{lemma:IntConvSet-characterization}, it suffices to show that
\[
\frac{x+y}{2}
\in
\conv\left(
S\cap N\left(\frac{x+y}{2}\right)
\right)
\]
for every $x,y\in S$.

Since $S$ is $\avgd$-closed, both $\avgd(x,y)$ and $\avgd(y,x)$
belong to $S$. By the definition of $\avgd$, they belong to
$
N\left(\frac{x+y}{2}\right)$,
and satisfy
$
\avgd(x,y)+\avgd(y,x)=x+y.
$
Therefore,
\[
\frac{x+y}{2}
=
\frac{\avgd(x,y)+\avgd(y,x)}{2}
\in
\conv\left(
S\cap N\left(\frac{x+y}{2}\right)
\right).
\]
Hence, $S$ is integrally convex.
\end{proof}

The following elementary inequality will be used to prove that UTVPI
systems are preserved by the directed midpoint operation.

\begin{lemma}
\label{lem:avgd-preserves-sum-inequality}
Let $a,b,c,d,r\in\mathbb Z$. If
$a+b\leq r$ and
$c+d\leq r$,
then
$\avgd(a,c)+\avgd(b,d)\leq r$.
\end{lemma}

\begin{proof}
If $a>c$ and $b>d$, then
$\avgd(a,c)\leq a$ and $\avgd(b,d)\leq b$, and the conclusion follows
from $a+b\leq r$.

Otherwise, at least one of the two directed midpoints is rounded
downward, while the other can be greater than its ordinary midpoint by
at most $1/2$. Hence,
\[
\avgd(a,c)+\avgd(b,d)
\leq
\frac{a+c}{2}+\frac{b+d}{2}+\frac12
\leq
r+\frac12.
\]
The left-hand side is an integer, and therefore it is at most $r$.
\end{proof}

\begin{lemma}
\label{prop:UTVPI-median-avg-closed}
Every UTVPI-representable set is closed under both $\avgd$ and the
median operation.
\end{lemma}

\begin{proof}
Let $S$ be represented by a UTVPI system $\mathcal L$.
Since every UTVPI system is a TVPI system, $S$ is
TVPI-representable and hence median-closed by
\Cref{prop:TVPI-is-median-closed}. It remains to prove that $S$ is
$\avgd$-closed.

Let
\[
\sigma x_i+\tau x_j\leq r
\]
be an inequality in $\mathcal L$, where
$\sigma,\tau\in\{0,-1,+1\}$. Since the left-hand side is
integer-valued on $\mathbb Z^V$, we may replace $r$ with
$\lfloor r\rfloor$ and assume that $r\in\mathbb Z$.

Let   $p,q\in S$. Then 
by 
applying \Cref{lem:avgd-preserves-sum-inequality} to
$
\sigma p_i+\tau p_j\leq r$ and 
$\sigma q_i+\tau q_j\leq r$, we obtain 
\[
\avgd(\sigma p_i,\sigma q_i)
+
\avgd( \tau p_j,\tau q_j)
\leq r.
\]
Since $\avgd(\sigma a,\sigma b)
=
\sigma\avgd(a,b)$ 
for all $\sigma\in\{0,-1,+1\}$, 
we have 
\[
\sigma\avgd(p_i,q_i)
+
\tau\avgd(p_j,q_j)
\leq r.
\]
Thus, $\avgd(p,q)$ satisfies every inequality in $\mathcal L$.
Therefore, $\avgd(p,q)\in S$, proving that  $S$ is $\avgd$-closed.
\end{proof}

We are now ready to prove \Cref{thm:main}.

\begin{proof}[Proof of \Cref{thm:main}]
The implication
\textup{(i)}~$\Rightarrow$~\textup{(ii)}
is \Cref{prop:UTVPI-median-avg-closed}. Since the median is a majority
operation, \textup{(ii)} implies \textup{(iii)}.

Suppose that \textup{(iii)} holds. By
\Cref{lem:avgd=>int-conv}, the set $S$ is integrally convex.
Moreover, closedness under a majority operation implies
$2$-decomposability by
\Cref{thm:near-unanimity-closed-is-decompo}. Thus,
\textup{(iv)} holds.

Now suppose that \textup{(iv)} holds. Integral convexity is preserved
under coordinate projections~\citep[Theorem~3.1]{MoM19}. Hence, for
every $W\in\binom{V}{2}$, the set $\proj_W(S)$ is integrally convex.
Since $|W|=2$, it is UTVPI-representable by
\Cref{prop:two-dim-IntConv=UTVPI}. Lifting the UTVPI representations
of these two-coordinate projections to $\mathbb Z^V$ and taking their
union yields a UTVPI system whose set of integer solutions is
\[
\Join_{W\in\binom{V}{2}}\proj_W(S)
=
S,
\]
where the equality follows from $2$-decomposability. Therefore,
\textup{(i)} holds.

Finally, the equivalence of \textup{(ii)} and \textup{(v)} was
established in
\Cref{cor:concrete-fixed-point-characterizations}\textup{(ii)}.
\end{proof}
The consequences of \Cref{thm:main} for intersection
characterizations are examined in detail in
\Cref{subsec:representability-hierarchy}. For $n\geq3$, we obtain
there the complete list of nontrivial pairwise intersection
characterizations of UTVPI-representability among the classes
considered in this paper. This result contrasts with
TVPI-representability, for which no analogous characterization exists
within the same collection of classes.

\subsection{Characterizations of SVPI-representability}
\label{subsec:SVPI-multioperation}

\providecommand{\mnc}{midpoint-neighborhood-closed}
\providecommand{\mncness}{midpoint-neighborhood-closedness}
\newcommand{\mnd}{\mathsf{mnd}}
\newcommand{\boxop}{\mathsf{box}}

The preceding results show that DC- and UTVPI-representable sets admit
characterizations by closedness under ordinary operations, that is,
the single-valued operations introduced in
\Cref{subsec-closure}. The situation is different for
SVPI-representability.
We first explain why ordinary operations do not suffice and then
characterize SVPI-representable sets by box-closedness and
\mncness, interpret these properties in terms of two binary
\emph{multioperations}, and derive a one-dimensional reconstruction formula.

\begin{lemma}
\label{lem:diagonal-closed-under-all-operations}
The diagonal set
\[
\Delta_V
=
\{\,(a,\ldots,a)\in\mathbb Z^V\mid a\in\mathbb Z\,\}
\]
is closed under every operation on $\mathbb Z$.
\end{lemma}

\begin{proof}
Let $f:\mathbb Z^k\to\mathbb Z$ be any operation, and let
$a_1,\ldots,a_k\in\mathbb Z$. Then
\[
f\bigl(
(a_1,\ldots,a_1),\ldots,(a_k,\ldots,a_k)
\bigr)
=
\bigl(
f(a_1,\ldots,a_k),\ldots,f(a_1,\ldots,a_k)
\bigr)
\in\Delta_V.
\]
\end{proof}

\begin{proposition}
\label{prop:SVPI-not-total-operation-characterizable}
There is no family $F$ of ordinary operations on $\mathbb Z$ that
characterizes SVPI-representability by closedness; that is, no such
$F$ satisfies
\[
S\text{ is SVPI-representable}
\quad\Longleftrightarrow\quad
S\text{ is }F\text{-closed}
\]
for every $S\subseteq\mathbb Z^V$.
\end{proposition}

\begin{proof}
By \Cref{lem:diagonal-closed-under-all-operations}, the diagonal set
$\Delta_V$ is $F$-closed for every family $F$ of operations.
However, it is not difficult to see that  $\Delta_V$ is not SVPI-representable. 
\end{proof}

For $x,y\in\mathbb Z^V$, define the \emph{box} spanned by $x$ and $y$
by
\[
[x,y]
=
\left\{
z\in\mathbb Z^V
\ \middle|\
\min\{x_i,y_i\}
\leq z_i\leq
\max\{x_i,y_i\}
\text{ for every }i\in V
\right\}.
\]
Throughout this subsection, $[x,y]$ denotes this integer box.
A set $S\subseteq\mathbb Z^V$ is called \emph{box-closed} if
$[x,y]\subseteq S$
for every $x,y\in S$.
It is called \emph{\mnc} if
$
N\left(\frac{x+y}{2}\right)\subseteq S$
for every $x,y\in S$.

Both properties above can be expressed in terms of multioperations.
Following \citep[Section~4]{Borner08}, a map of the form
\[
f:A^k\longrightarrow 2^A
\]
is called a \emph{multioperation} or a \emph{multifunction}.
In this paper, we restrict attention to \emph{total} multioperations,
that is, multioperations satisfying
\[
f(\boldsymbol{a})\neq\emptyset
\qquad
\text{for every }\boldsymbol{a}\in A^k.
\]

Define the binary multioperations
$\boxop,\mnd:\mathbb Z^2\to 2^{\mathbb Z}$ by
\begin{eqnarray*}
\boxop(a,b)
&=&
\left\{
c\in\mathbb Z
\ \middle|\
\min\{a,b\}\leq c\leq\max\{a,b\}
\right\}\\
\mnd(a,b)
&=&
\left\{
c\in\mathbb Z
\ \middle|\
\left|c-\frac{a+b}{2}\right|<1
\right\}.
\end{eqnarray*}
Their componentwise extensions satisfy
\[
\boxop(x,y)=[x,y]
\ \ \  \text{and}\ \ \ 
\mnd(x,y)
=
N\left(\frac{x+y}{2}\right)
\ \ \  
\text{for all } x,y \in \mathbb{Z}^V.
\]
Thus, $S$ is box-closed if and only if $\boxop$ preserves $S$, whereas
$S$ is \mnc\ if and only if $\mnd$ preserves $S$. In analogy with
ordinary polymorphisms, we call such a preserving multioperation a
\emph{multipolymorphism} of $S$.

Since
\[
\mnd(x,y)\subseteq \boxop(x,y)
\]
for every $x,y\in\mathbb Z^V$, box-closedness immediately implies
\mncness. The converse is the nontrivial direction.

\begin{lemma}[Midpoint neighborhoods generate boxes]
\label{lem:M-closed-box}
Every \mnc\ set is box-closed.
\end{lemma}

\begin{proof}
We first establish a unit-step property. Let $S$ be \mnc, let
$x,y\in S$, and fix $i\in V$ with $x_i\neq y_i$. Set
\[
\varepsilon=\operatorname{sgn}(y_i-x_i).
\]
We show that
\[
x+\varepsilon e_i\in S,
\]
where $e_i$ is the $i$th unit vector.

Set $y^{(0)}=y$. Given $y^{(t)}\in S$, choose
\[
y^{(t+1)}
\in
N\left(\frac{x+y^{(t)}}{2}\right)
\]
coordinatewise as follows. For every $j\neq i$, choose
$y_j^{(t+1)}$ so that its distance from $x_j$ is as small as possible.
For $j=i$, make the same choice while
$|y_i^{(t)}-x_i|>1$; once this distance becomes $1$, choose
$y_i^{(t+1)}=y_i^{(t)}$. Such a choice is always possible. Since
$S$ is \mnc, every $y^{(t)}$ belongs to $S$.

For $j\neq i$, every positive distance
$|y_j^{(t)}-x_j|$ strictly decreases. The $i$th distance strictly
decreases until it becomes $1$ and then remains equal to $1$.
Consequently, after finitely many steps,
\[
y^{(t)}=x+\varepsilon e_i,
\]
which proves the unit-step property.

Now let $z\in[x,y]$. Starting from $x$, repeatedly choose a coordinate
in which the current vector differs from $z$ and apply the unit-step
property to the current vector and $y$. Each step changes only the
chosen coordinate by one unit toward $z$. The procedure reaches $z$
after finitely many steps, and every intermediate vector belongs to
$S$. Hence, $z\in S$, proving that $S$ is box-closed.
\end{proof}

\begin{theorem}[Characterizations of SVPI representability]
\label{thm:SVPI-neighborhood-characterization}
For a set $S\subseteq\mathbb Z^V$, the following statements are
equivalent:
\begin{enumerate}
\setlength{\itemsep}{4pt}
\setlength{\parsep}{0pt}

\item[\textup{(i)}]
$S$ is SVPI-representable.

\item[\textup{(ii)}]
$S$ is box-closed.

\item[\textup{(iii)}]
$S$ is \mnc.

\item[\textup{(iv)}]
$
S
=
\prod_{i\in V}
\left(
\conv\bigl(\proj_{\{i\}}(S)\bigr)\cap\mathbb Z
\right).
$
\end{enumerate}
\end{theorem}

\begin{proof}
The equivalence of \textup{(i)}, \textup{(ii)}, and \textup{(iv)}
follows from the elementary fact that the integer solution sets of
SVPI systems are precisely the Cartesian products of integer
intervals. Indeed, box-closedness forces each one-coordinate
projection to be an integer interval and permits arbitrary
coordinatewise combinations of these projections, yielding exactly
the product in \textup{(iv)}.

The implication \textup{(ii)} $\Rightarrow$ \textup{(iii)} follows
from
\[
N\left(\frac{x+y}{2}\right)\subseteq[x,y]
\qquad
\text{for every }x,y\in\mathbb Z^V.
\]
The converse is \Cref{lem:M-closed-box}.
\end{proof}

\begin{remark}
Condition \textup{(iv)} is the one-dimensional analogue of the
pairwise closure-and-join formulas in
\Cref{sec:2-decomp-closure}: each one-coordinate projection is replaced
by its integer convex hull, and the resulting one-dimensional sets are
recombined by taking their Cartesian product.
\end{remark}

\subsection{Limits of operation-based characterizations}
\label{subsec:limits-operation-characterizations}

The preceding results provide operation-based characterizations of
DC- and UTVPI-representability by ordinary operations and of
SVPI-representability by multioperations. 

Among the classes considered
in this paper, the five remaining classes for which no such
characterization has yet been given are the TVPI-representable,
closed-hole-free, hole-free, $2$-decomposable, and integrally convex
sets.
We now determine whether these classes admit analogous
characterizations.
We show that the first four classes admit no operation-based
characterization that is uniform over all dimensions. For integrally
convex sets, no such characterization exists even in any fixed
dimension $n\geq3$.
Here, a characterization by multioperations is called
\emph{uniform over all dimensions} if there exists a family $F$ of
total multioperations on $\mathbb Z$, independent of the dimension
$n$, such that, for every $n$ and every
$S\subseteq\mathbb Z^n$, the set $S$ belongs to the class under
consideration if and only if it is $F$-closed.

We first note two basic properties of closedness under
multioperations, which apply in particular to ordinary operations.
Recall that, throughout this paper, every multioperation is assumed
to be total; that is, a $k$-ary multioperation on $\mathbb Z$ is a map
\[
f:\mathbb Z^k\longrightarrow
2^{\mathbb Z}\setminus\{\emptyset\}.
\]

\begin{lemma}
\label{lem:preserved-classes-projection-intersection}
Let $F$ be a family of  multioperations on $\mathbb Z$.
Then the following statements hold:
\begin{enumerate}
\setlength{\itemsep}{4pt}
\setlength{\parsep}{0pt}

\item[\textup{(i)}]
If $\mathcal S$ is a family of $F$-closed subsets of
$\mathbb Z^V$, then
$
\bigcap_{T\in\mathcal T}T
$
is $F$-closed for every subfamily
$\mathcal T\subseteq\mathcal S$.

\item[\textup{(ii)}]
If $S\subseteq\mathbb Z^V$ is $F$-closed, then
$\proj_W(S)$ is $F$-closed for every $W\subseteq V$.
\end{enumerate}
\end{lemma}

\begin{proof}
Statement \textup{(i)} follows directly from the definition of
$F$-closedness. For \textup{(ii)}, the proof is the same lifting
argument as in \Cref{lemma:proj-closed}\textup{(i)}.
The only additional point is that totality allows an arbitrary output
on $W$ to be extended to a full output vector, which belongs to $S$
by $F$-closedness; hence, the original output belongs to
$\proj_W(S)$.
\end{proof}

The preceding lemma reduces the desired impossibility results to
failures of closure under intersections and coordinate projections.
We establish the required failures by the examples below.

\begin{theorem}[Limits of operation-based characterizations]
\label{thm:limits-operation-characterizations}
The following statements hold:
\begin{enumerate}
\setlength{\itemsep}{4pt}
\setlength{\parsep}{0pt}

\item[\textup{(i)}]
None of the classes of TVPI-representable, closed-hole-free,
hole-free, or $2$-decomposable sets admits a characterization by
multioperations that is uniform over all dimensions.

\item[\textup{(ii)}]
For every fixed $n\geq3$, there is no family $F$ of 
multioperations on $\mathbb Z$ such that
\[
S\text{ is integrally convex}
\quad\Longleftrightarrow\quad
S\text{ is }F\text{-closed}
\]
holds for every $S\subseteq\mathbb Z^n$.
\end{enumerate}
\end{theorem}

\begin{proof}
We first prove \textup{(i)} for the classes of TVPI-representable,
closed-hole-free, and hole-free sets. Let $S$ be the set in
\Cref{rem:closed-hole-free-not-preserved-by-projection}.
As shown there, $S$ is closed-hole-free and hence hole-free, whereas
$\proj_{\{1,2\}}(S)$ is neither closed-hole-free nor hole-free.

Moreover, $S$ is TVPI-representable. Indeed, $S$ is the set of
integer solutions to the TVPI system
\[
0\leq x_1\leq2,
\qquad
x_1=x_2,
\qquad
3x_1=2x_3,
\]
where each equality is represented by two inequalities. On the other
hand, $\proj_{\{1,2\}}(S)$ is not TVPI-representable. Indeed, every
TVPI-representable set is linearly representable and hence
closed-hole-free by \Cref{prop:linear-holefree}. Thus, none of the
classes of TVPI-representable, closed-hole-free, or hole-free sets is
preserved under coordinate projections. By
\Cref{lem:preserved-classes-projection-intersection}~\textup{(ii)},
none of these three classes admits a characterization by multioperations that is uniform over all dimensions.

We next prove \textup{(i)} for the class of $2$-decomposable sets.
Let
\[
T
=
\{\,
(0,0,1,0),\,
(0,1,0,1),\,
(1,0,0,2)
\,\}
\subseteq
\mathbb Z^4.
\]
The fourth coordinates of the three vectors are pairwise distinct.
If
\[
x
\in
\Join_{W\in\binom{\{1,2,3,4\}}{2}}\proj_W(T),
\]
then $x_4\in\{0,1,2\}$. Once $x_4$ is fixed, the pairwise
projections involving the fourth coordinate uniquely determine all
the other coordinates of $x$. Hence, $x\in T$, proving that $T$ is
$2$-decomposable.

However,
\[
\proj_{\{1,2,3\}}(T)
=
\{\,(0,0,1),(0,1,0),(1,0,0)\,\}
\]
is not $2$-decomposable. Indeed, each of its two-coordinate
projections contains $(0,0)$, and hence $(0,0,0)$ belongs to the join
of these projections, whereas
\[
(0,0,0)
\notin
\proj_{\{1,2,3\}}(T).
\]
Thus, the class of $2$-decomposable sets is not preserved under
coordinate projections. Again by
\Cref{lem:preserved-classes-projection-intersection}~\textup{(ii)},
this class admits no characterization by multioperations that
is uniform over all dimensions.

We finally prove \textup{(ii)}. The following example is taken from
\citep[Example~4.4]{MuS01}. Let
\[
\begin{aligned}
I_1
&=
\{\,
(0,0,0),\,
(0,1,1),\,
(1,1,0),\,
(1,2,1)
\,\},\\
I_2
&=
\{\,
(0,0,0),\,
(0,1,0),\,
(1,1,1),\,
(1,2,1)
\,\}.
\end{aligned}
\]
As shown in \citep[Example~2.10]{MT23}, both $I_1$ and $I_2$ are
integrally convex,
whereas
\[
I_1\cap I_2
=
\{\,(0,0,0),(1,2,1)\,\}
\]
is not integrally convex.
These assertions can also be verified directly using
\Cref{lemma:IntConvSet-characterization}.

For $n>3$, define the zero extensions
\[
I_k^{(n)}
=
I_k\times\{0\}^{n-3}
\qquad
(k=1,2).
\]
It follows immediately from the definition that a set is integrally
convex if and only if its zero extension is integrally convex.
Therefore, $I_1^{(n)}$ and $I_2^{(n)}$ are integrally convex, whereas
\[
I_1^{(n)}\cap I_2^{(n)}
=
(I_1\cap I_2)\times\{0\}^{n-3}
\]
is not integrally convex.
Hence, for every $n\geq3$, the class of integrally convex subsets of
$\mathbb Z^n$ is not closed under intersections. Statement
\textup{(ii)} now follows from
\Cref{lem:preserved-classes-projection-intersection}~\textup{(i)}.
\end{proof}

The dimensional restriction in
\Cref{thm:limits-operation-characterizations}~\textup{(ii)} is sharp.
In dimension two, integrally convex sets coincide with
UTVPI-representable sets by
\Cref{prop:two-dim-IntConv=UTVPI}. Hence, by \Cref{thm:main}, a set
$S\subseteq\mathbb Z^2$ is integrally convex if and only if it is
closed under both $\avgd$ and the median operation.

Together with the positive characterizations established in the
preceding subsections, the theorem clarifies the extent to which the
principal classes considered in this paper can be described by
preservation under operations or multioperations.

\section{Relations among representability, convexity, and closure classes}
\label{sec:class-hierarchy}
\newcommand{\Dmp}{\mathsf{DMP}}

This section summarizes the relations among the classes introduced in
the preceding sections. These include classes defined by inequality
representability, discrete convexity, and closedness under operations.
We consider the same collection of classes in the general, Boolean,
and two-dimensional settings, and describe the inclusions,
equivalences, and incomparabilities among them.

For brevity, we use
$\mathsf{SVPI}$, $\mathsf{DC}$, $\mathsf{UTVPI}$,
$\mathsf{TVPI}$, and $\mathsf{Linear}$ to denote the corresponding
representability classes. We write $\Dmp$, $\mathsf{IC}$,
$\mathsf{CHF}$, $\mathsf{HF}$, $\mathsf{Median}$,
$\mathsf{Majority}$, $\mathsf{2Dec}$, and $\mathsf{All}$ for the
classes of $\avgd$-closed, integrally convex, closed-hole-free,
hole-free, median-closed, majority-closed, $2$-decomposable, and
arbitrary sets in the domain under
consideration, respectively. Here,
$\mathsf{Majority}$ denotes the class of sets that are closed under
at least one majority operation.

\subsection{The general setting}
\label{subsec:representability-hierarchy}

The following proposition summarizes the inclusion hierarchy.
We prove below all inclusions represented by the arrows in 
\Cref{fig:representability-hierarchy}. For
$n\geq3$, the strictness of these inclusions and the absence of any
further inclusion relations are established by the separating
examples in \ref{app:hierarchy-witnesses}.

\begin{figure}[t]
\centering
\begin{tikzpicture}[
  every node/.style={
    font=\small,
    align=center,
    inner sep=3pt
  },
  implication/.style={
    double equal sign distance,
    -{Implies[length=3mm]}
  },
  x=1cm,
  y=1cm
]
\node (all) at (0,8.2) {
  $\mathsf{All}$
};

\node (hf) at (-4.2,7.0) {
  $\mathsf{HF}$
};
\node (decomp) at (4.2,7.0) {
  $\mathsf{2Dec}$
};

\node (chf) at (-4.2,5.8) {
  $\mathsf{CHF}=\mathsf{Linear}$
};
\node (majority) at (4.2,5.8) {
  $\mathsf{Majority}$
};

\node (ic) at (-4.8,4.5) {
  $\mathsf{IC}$
};
\node (median) at (3.1,4.5) {
  $\mathsf{Median}$
};

\node (dmp) at (-4.8,3.2) {
  $\Dmp$
};
\node (tvpi) at (0.8,3.5) {
  $\mathsf{TVPI}$
};

\node (utvpi) at (-1.4,2.0) {
  $\mathsf{UTVPI}$
};

\node (dc) at (-1.4,0.8) {
  $\mathsf{DC}$
};

\node (svpi) at (-1.4,-0.5) {
  $\mathsf{SVPI}$
};

\draw[implication] (svpi) -- (dc);
\draw[implication] (dc) -- (utvpi);
\draw[implication] (utvpi) -- (dmp);
\draw[implication] (utvpi) -- (tvpi);
\draw[implication] (dmp) -- (ic);
\draw[implication] (ic) -- (chf);
\draw[implication] (tvpi) -- (chf);
\draw[implication] (tvpi) -- (median);
\draw[implication] (median) -- (majority);
\draw[implication] (majority) -- (decomp);
\draw[implication] (chf) -- (hf);
\draw[implication] (hf) -- (all);
\draw[implication] (decomp) -- (all);
\end{tikzpicture}
\caption{Hasse diagram of the inclusion hierarchy among the classes of
subsets of $\mathbb Z^n$ for $n\geq3$}
\label{fig:representability-hierarchy}
\end{figure}

\begin{proposition}[Inclusion hierarchy in the general setting]
\label{prop:representability-hierarchy}
For every $n>1$, all inclusions displayed in
\Cref{fig:representability-hierarchy} hold. For $n\geq3$, the figure
is the Hasse diagram of the displayed classes under inclusion.\footnote{Each upward arrow
represents a proper inclusion and the only inclusion relations are those implied by directed paths in the diagram.} 
\end{proposition}

\begin{proof}
The inclusions
\[
\mathsf{SVPI}
\subseteq
\mathsf{DC}
\subseteq
\mathsf{UTVPI}
\subseteq
\mathsf{TVPI}
\]
follow directly from the definitions of the corresponding systems.

Every UTVPI-representable set is $\avgd$-closed by
\Cref{prop:UTVPI-median-avg-closed}, and hence
\[
\mathsf{UTVPI}
\subseteq
\Dmp
\subseteq
\mathsf{IC} \subseteq
\mathsf{CHF},
\]
where the second and third inclusions follow from
\Cref{lem:avgd=>int-conv} and 
\eqref{eq-convexrelations}, respectively. 
By \Cref{prop:linear-holefree},
\[
\mathsf{CHF}=\mathsf{Linear},
\]
and the inclusions
\[
\mathsf{CHF}
\subseteq
\mathsf{HF}
\subseteq
\mathsf{All}
\]
follow from the definitions.

Every TVPI-representable set is linearly representable and hence
closed-hole-free. It is also median-closed by
\Cref{prop:TVPI-is-median-closed}. Since the median is a majority
operation,
\[
\mathsf{Median}
\subseteq
\mathsf{Majority},
\]
and every set closed under a majority operation is $2$-decomposable by
\Cref{thm:near-unanimity-closed-is-decompo}. Finally,
\[
\mathsf{2Dec}\subseteq\mathsf{All}
\]
is immediate. This proves all the inclusions represented by the arrows
in \Cref{fig:representability-hierarchy}.

For $n\geq3$, the separating examples in
\ref{app:hierarchy-witnesses} show that every arrow is strict and
that all inclusion relations among the displayed classes are exactly
those implied by directed paths. Hence,
\Cref{fig:representability-hierarchy} is the Hasse diagram of the
displayed classes.
\end{proof}

Thus, for $n\geq3$,
\Cref{prop:representability-hierarchy} determines both all inclusions
and all non-inclusions among the displayed classes. We now use this
complete hierarchy to classify all nontrivial pairwise intersection
characterizations of $\mathsf{UTVPI}$.

\begin{corollary}[Complete list of nontrivial pairwise intersection
characterizations of UTVPI]
\label{cor:complete-UTVPI-intersections}
For $n\geq3$, let $\mathsf{A}$ and $\mathsf{B}$ be two of the classes 
displayed in \Cref{fig:representability-hierarchy}, 
neither of which
is equal to $\mathsf{UTVPI}$. Then
\[
\mathsf{A}\cap\mathsf{B}
=
\mathsf{UTVPI}
\]
if and only if, after possibly interchanging $\mathsf{A}$ and
$\mathsf{B}$,
\[
\mathsf{A}\in\{\Dmp,\mathsf{IC}\}
\quad\text{and}\quad
\mathsf{B}\in
\{
\mathsf{TVPI},
\mathsf{Median},
\mathsf{Majority},
\mathsf{2Dec}
\}.
\]
\end{corollary}
\begin{proof}
We first prove the if-direction. 
By
\Cref{prop:representability-hierarchy}, whose strictness assertion is
established by the separating examples in
\ref{app:hierarchy-witnesses},  we have
\[
\mathsf{UTVPI}
\subsetneq
\Dmp
\subsetneq
\mathsf{IC}
\]
and
\[
\mathsf{UTVPI}
\subsetneq
\mathsf{TVPI}
\subsetneq
\mathsf{Median}
\subsetneq
\mathsf{Majority}
\subsetneq
\mathsf{2Dec}.
\]
Suppose, after possibly interchanging $\mathsf{A}$ and $\mathsf{B}$,
that the right-hand condition in the statement of the corollary
holds. Then
\[
\mathsf{UTVPI}
\subsetneq
\mathsf{A}
\subseteq
\mathsf{IC}
\ \mbox{ and } \
\mathsf{UTVPI}
\subsetneq
\mathsf{B}
\subseteq
\mathsf{2Dec}.
\]
Therefore,
\[
\mathsf{UTVPI}
\subseteq
\mathsf{A}\cap\mathsf{B}
\subseteq
\mathsf{IC}\cap\mathsf{2Dec}
=
\mathsf{UTVPI},
\]
where the equality follows from \Cref{thm:main}. Hence,
\[
\mathsf{A}\cap\mathsf{B}
=
\mathsf{UTVPI}.
\]

We next prove the ``only-if'' direction. Suppose that
\[
\mathsf{A}\cap\mathsf{B}
=
\mathsf{UTVPI}.
\]
Then
\[
\mathsf{A}\supseteq\mathsf{UTVPI}
\qquad\text{and}\qquad
\mathsf{B}\supseteq\mathsf{UTVPI}.
\]
Since neither $\mathsf{A}$ nor $\mathsf{B}$ is equal to
$\mathsf{UTVPI}$, both are proper superclasses of
$\mathsf{UTVPI}$.

We now use the completeness of the Hasse diagram in
\Cref{prop:representability-hierarchy}, established by the separating
examples in \ref{app:hierarchy-witnesses}. 
Every proper displayed
superclass of $\mathsf{UTVPI}$ other than $\Dmp$ and $\mathsf{IC}$
contains $\mathsf{TVPI}$. Hence, if neither $\mathsf{A}$ nor
$\mathsf{B}$ belongs to
$\{\Dmp,\mathsf{IC}\}$, then
\[
\mathsf{A}\cap\mathsf{B}
\supseteq
\mathsf{TVPI}
\supsetneq
\mathsf{UTVPI},
\]
a contradiction. Thus, after possibly interchanging $\mathsf{A}$ and
$\mathsf{B}$, we may assume that
\[
\mathsf{A}\in\{\Dmp,\mathsf{IC}\}.
\]

Suppose now that
\[
\mathsf{B}
\notin
\{
\mathsf{TVPI},
\mathsf{Median},
\mathsf{Majority},
\mathsf{2Dec}
\}.
\]
Since $\mathsf{B}$ is a proper displayed superclass of
$\mathsf{UTVPI}$, the complete Hasse diagram shows that
$\mathsf{A}$ and $\mathsf{B}$ are comparable. Their intersection is
therefore  a proper
superclass of $\mathsf{UTVPI}$,  
again a contradiction. Therefore,
\[
\mathsf{B}\in
\{
\mathsf{TVPI},
\mathsf{Median},
\mathsf{Majority},
\mathsf{2Dec}
\}.
\]
This completes the proof. 
\end{proof}

\begin{remark}[Extremal intersection characterizations of UTVPI]
\label{rem:extremal-UTVPI-intersections}
By \Cref{cor:complete-UTVPI-intersections}, the preceding list
exhausts all nontrivial pairwise intersection characterizations of
$\mathsf{UTVPI}$ among the displayed classes. Since
\[
\Dmp
\subsetneq
\mathsf{IC}
\ \mbox{ and } \
\mathsf{TVPI}
\subsetneq
\mathsf{Median}
\subsetneq
\mathsf{Majority}
\subsetneq
\mathsf{2Dec},
\]
the characterization
\[
\mathsf{UTVPI}
=
\mathsf{IC}\cap\mathsf{2Dec}
\]
uses the largest possible class from each of the two branches, whereas
\[
\mathsf{UTVPI}
=
\Dmp\cap\mathsf{TVPI}
\]
uses the smallest possible class from each branch. Equivalently, the
former combines the least restrictive pair of conditions, while the
latter combines the most restrictive pair.
\end{remark}

\begin{remark}[Failure of an analogous characterization for TVPI]
\label{rem:sharpness-UTVPI-intersection}
In contrast, no analogous nontrivial pairwise intersection
characterization exists for $\mathsf{TVPI}$ among the displayed
classes.
Its two immediate
proper superclasses are $\mathsf{Linear}$ and $\mathsf{Median}$, but
their intersection strictly contains $\mathsf{TVPI}$, as shown by the
following example.
\end{remark}

\begin{example}[Linearly representable and median-closed but not
TVPI-representable]
\label{ex:linear-median-not-TVPI}
Let
\[
L_0
=
\{
a=(0,0,0),\
b=(0,1,1),\
c=(2,1,2)
\}
\subseteq\mathbb Z^3.
\]
We show that
\[
L_0
\in
\mathsf{Linear}\cap\mathsf{Median}
\setminus
\mathsf{TVPI}.
\]

First, $L_0$ is linearly representable; indeed,
\[
L_0
=
\left\{
x\in\mathbb Z^3
\ \middle|\
0\leq x_1\leq 2x_2\leq 2,\quad
x_1+2x_2-2x_3=0
\right\}.
\]
The equality may be replaced by two linear inequalities.
These constraints imply $0\leq x_2\leq1$. If $x_2=0$, then
$x_1=x_3=0$, yielding $a$.  If $x_2=1$, then $0\leq x_1\leq2$ and
\[
x_1+2=2x_3.
\]
Since $x_3$ is an integer, $x_1$ must be even. Hence,
$x_1\in\{0,2\}$, yielding $b$ and $c$, respectively. Thus, the
displayed system has exactly $L_0$ as its set of integer solutions.

Moreover, $L_0$ is median-closed. Every triple of points from $L_0$
with a repeated point has the repeated point as its median, while
\[
\median(a,b,c)=b.
\]

It remains to show that $L_0$ is not TVPI-representable. By
\Cref{thm:pairwise-closed-convex-hull}, it suffices to show that
\[
L_0\neq\Psi(L_0).
\]
Consider the vector
$
z=(1,1,1)\notin L_0.
$
Then
\[
\begin{aligned}
\proj_{\{1,2\}}(z)
&=
\frac12\proj_{\{1,2\}}(b)
+
\frac12\proj_{\{1,2\}}(c),\\
\proj_{\{1,3\}}(z)
&=
\frac12\proj_{\{1,3\}}(a)
+
\frac12\proj_{\{1,3\}}(c),\\
\proj_{\{2,3\}}(z)
&=
\proj_{\{2,3\}}(b).
\end{aligned}
\]
Thus, every two-coordinate projection of $z$ belongs to the closed
convex hull of the corresponding projection of $L_0$. Hence,
$z\in\Psi(L_0)$, and therefore
\[
L_0\neq\Psi(L_0).
\]
By \Cref{thm:pairwise-closed-convex-hull}, $L_0$ is not
TVPI-representable. Consequently,
\[
L_0
\in
(\mathsf{Linear}\cap\mathsf{Median})
\setminus
\mathsf{TVPI}.
\]

For $n>3$, the zero extension
\[
L_0\times\{0\}^{n-3}
\]
is again linearly representable and median-closed. It is not
TVPI-representable either: if it had a TVPI representation, setting
the last $n-3$ variables to zero would yield a TVPI representation of
$L_0$. Since every TVPI-representable set is linearly representable
and median-closed, it follows that, for every $n\geq3$,
\[
\mathsf{TVPI}
\subsetneq
\mathsf{Linear}\cap\mathsf{Median}.
\]

As shown in
\Cref{fig:representability-hierarchy}, $\mathsf{Linear}$ and
$\mathsf{Median}$ are the two immediate proper superclasses of
$\mathsf{TVPI}$. Every proper displayed superclass of
$\mathsf{TVPI}$ therefore contains at least one of these two classes,
and hence contains their intersection. Thus, every intersection of
proper displayed superclasses contains
$\mathsf{Linear}\cap\mathsf{Median}$ and is strictly larger than
$\mathsf{TVPI}$. Consequently, no such intersection characterizes
$\mathsf{TVPI}$.
\end{example}

\subsection{The Boolean setting}
\label{subsec:boolean-setting}

We next specialize the same classes to the Boolean domain. The
resulting inclusion hierarchy is shown in
\Cref{fig:inclusion-Boolean}.

\begin{proposition}[Hierarchy in the Boolean setting]
\label{prop:hierarchy-in-Boolean}
For a set $S\subseteq\{0,1\}^n$,  the classes considered above collapse
as follows:
\begin{align}
\mathsf{All}
&=
\mathsf{HF}
=
\mathsf{CHF}
=
\mathsf{Linear}
=
\mathsf{IC}
=
\Dmp,\label{eq-boolclass1}\\
\mathsf{2Dec}
&=
\mathsf{Majority}
=
\mathsf{Median}
=
\mathsf{TVPI}
=
\mathsf{UTVPI}.\label{eq-boolclass2}
\end{align}
Moreover,
\[
\mathsf{SVPI}
\subsetneq
\mathsf{DC}
\subsetneq
\mathsf{2Dec}
\subseteq
\mathsf{All}.
\]
The last inclusion is strict for $n\geq3$ and is an equality for
$n=2$.
\end{proposition}

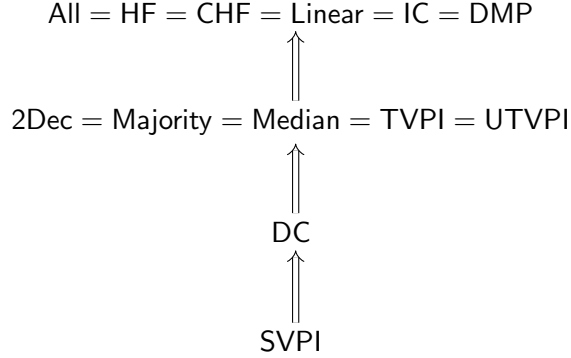
\begin{figure}[tb]
\centering
\begin{tikzpicture}[
  every node/.style={
    font=\small,
    align=center,
    inner sep=3pt
  },
  implication/.style={
    double equal sign distance,
    -{Implies[length=3mm]}
  },
  x=1cm,
  y=1.45cm
]
\node (all) at (0,3) {
  $\mathsf{All}
   =\mathsf{HF}   
   =\mathsf{CHF}
   =\mathsf{Linear}
   =\mathsf{IC}
   =\Dmp$
};

\node (two) at (0,2) {
  $\mathsf{2Dec}
  =\mathsf{Majority}
  =\mathsf{Median}
   =\mathsf{TVPI}
    =\mathsf{UTVPI}
   $
};

\node (dc) at (0,1) {
  $\mathsf{DC}$
};

\node (svpi) at (0,0) {
  $\mathsf{SVPI}$
};

\draw[implication] (svpi) -- (dc);
\draw[implication] (dc) -- (two);
\draw[implication] (two) -- (all);
\end{tikzpicture}
\caption{Hasse diagram of the inclusion hierarchy among the classes of
subsets of $\{0,1\}^n$ for $n\geq3$.  When $n=2$, the top two nodes coincide.}
\label{fig:inclusion-Boolean}
\end{figure}

\begin{proof}
For $a,b\in\{0,1\}$, we have
\[
\avgd(a,b)=a.
\]
Thus, $\avgd(x,y)=x$ for all $x,y\in\{0,1\}^n$, and hence every
subset of $\{0,1\}^n$ is $\avgd$-closed. Together with
\Cref{prop:representability-hierarchy}, this proves the equalities in
\eqref{eq-boolclass1}.

Next, let $S\subseteq\{0,1\}^n$ be $2$-decomposable. 
We first impose the unary bounds
\[
0\leq x_i\leq1
\qquad
(i\in V).
\]
For each $W=\{i,j\}$ with $1\leq i<j\leq n$ and each
\[
(a,b)\in\{0,1\}^2\setminus\proj_W(S),
\]
consider the inequality
\[
(2a-1)x_i+(2b-1)x_j\leq a+b-1.
\]
Then for Boolean $x_i$ and $x_j$, this inequality holds if and only if
\[
(x_i,x_j)\neq(a,b).
\]
Therefore, the set of Boolean solutions to these inequalities is
\[
\Join_{W\in\binom{[n]}{2}}\proj_W(S)=S,
\]
where $[n] = \{1,\dots, n\}$.
Thus, $S$ is UTVPI-representable. Together with
\Cref{prop:representability-hierarchy}, this proves the equalities in
\eqref{eq-boolclass2}.

For $n=2$, every subset $S$ of $\{0,1\}^2$ is $2$-decomposable, and
hence
\[
\mathsf{2Dec}=\mathsf{All}.
\]
It follows that the two nodes corresponding to
\eqref{eq-boolclass1} and \eqref{eq-boolclass2} coincide.

Finally, the examples in the appendix show that every arrow other
than the top one represents a proper inclusion for $n\geq2$, whereas
the top arrow represents a proper inclusion for $n\geq3$.
\end{proof}

\subsection{The two-dimensional setting}
\label{rmk:two-dim-01-characterization}
Finally, we specialize the same classes to subsets of $\mathbb Z^2$.
The resulting inclusion hierarchy is shown in
\Cref{fig:inclusion-two-dim}.

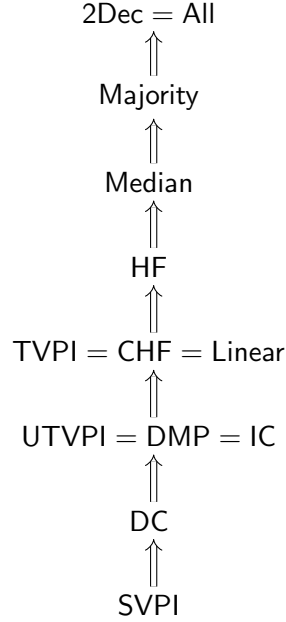
\begin{figure}[tb]
\centering
\begin{tikzpicture}[
  every node/.style={
    font=\small,
    align=center,
    inner sep=3pt
  },
  implication/.style={
    double equal sign distance,
    -{Implies[length=3mm]}
  },
  x=1cm,
  y=1.12cm
]
\node (all) at (0,7) {
  $\mathsf{2Dec}=\mathsf{All}$
};

\node (majority) at (0,6) {
  $\mathsf{Majority}$
};

\node (median) at (0,5) {
  $\mathsf{Median}$
};

\node (hf) at (0,4) {
  $\mathsf{HF}$
};

\node (tvpi) at (0,3) {
  $\mathsf{TVPI}=\mathsf{CHF}=\mathsf{Linear}$
};

\node (utvpi) at (0,2) {
  $\mathsf{UTVPI}=\Dmp=\mathsf{IC}$
};

\node (dc) at (0,1) {
  $\mathsf{DC}$
};

\node (svpi) at (0,0) {
  $\mathsf{SVPI}$
};

\draw[implication] (svpi) -- (dc);
\draw[implication] (dc) -- (utvpi);
\draw[implication] (utvpi) -- (tvpi);
\draw[implication] (tvpi) -- (hf);
\draw[implication] (hf) -- (median);
\draw[implication] (median) -- (majority);
\draw[implication] (majority) -- (all);
\end{tikzpicture}
\caption{Hasse diagram of the inclusion hierarchy among the classes of
subsets of $\mathbb Z^2$.}
\label{fig:inclusion-two-dim}
\end{figure}

\begin{proposition}[Complete inclusion hierarchy in dimension two]
\label{prop:hierarchy-in-2-dim}
For subsets of $\mathbb Z^2$, the complete inclusion hierarchy among
the classes considered above is given by
\[
\begin{aligned}
\mathsf{SVPI}
&\subsetneq
\mathsf{DC}
\subsetneq
\mathsf{UTVPI}
=
\Dmp
=
\mathsf{IC}\\
&\subsetneq
\mathsf{TVPI}
=
\mathsf{CHF}
=
\mathsf{Linear}
\subsetneq
\mathsf{HF}\\
&\subsetneq
\mathsf{Median}
\subsetneq
\mathsf{Majority}
\subsetneq
\mathsf{2Dec}
=
\mathsf{All}.
\end{aligned}
\]
\end{proposition}

\begin{proof}
By \Cref{prop:representability-hierarchy}, it remains only to establish
the following additional inclusions specific to dimension two:
\[
\mathsf{IC}
\subseteq
\mathsf{UTVPI},
\qquad
\mathsf{Linear}
\subseteq
\mathsf{TVPI},
\qquad
\mathsf{HF}
\subseteq
\mathsf{Median},
\qquad
\mathsf{All}
\subseteq
\mathsf{2Dec}.
\]

The first inclusion follows from
\Cref{prop:two-dim-IntConv=UTVPI}.
The second follows because every linear inequality in two variables
is a TVPI inequality.

To prove the third inclusion, let $S\subseteq\mathbb Z^2$ be
hole-free, and let $x,y,z\in S$. Since
$\conv\{x,y,z\}$ is a polytope, the set
\[
\conv\{x,y,z\}\cap\mathbb Z^2
\]
is linearly representable and hence TVPI-representable. It is therefore
median-closed by \Cref{prop:TVPI-is-median-closed}. Consequently,
\[
\median(x,y,z)
\in
\conv\{x,y,z\}\cap\mathbb Z^2
\subseteq
\conv(S)\cap\mathbb Z^2
=
S.
\]
Thus, $S$ is median-closed.

Finally, every subset $S$ of $\mathbb Z^2$ is $2$-decomposable, since
\[
\Join_{W\in\binom{[2]}{2}}\proj_W(S)
=
\proj_{\{1,2\}}(S)
=
S.
\]
Combining these inclusions with
\Cref{prop:representability-hierarchy} yields all the equalities and
inclusions in the displayed hierarchy.

The seven proper inclusions displayed above are witnessed by the
examples in the appendix.
\end{proof}

\section*{Acknowledgments}
The authors thank Akihisa Tamura for helpful comments and suggestions that improved the presentation of the paper.
The first author was partially supported by JST, ACT-X Grant Number JPMJAX200C and ERATO Grant Number JPMJER2301, and JSPS KAKENHI Grant Number JP21K17700, Japan.
The third author was supported by JST CREST Grant Number JPMJCR22M1, Japan.
The fourth author was supported by WISE program (MEXT) at Kyushu University.
This work was partially supported by the joint project of Kyoto University and Toyota Motor Corporation.

 \bibliographystyle{abbrvnat}
 \bibliography{UTVPI}
 



\appendix

\section{Separating examples for the inclusion hierarchies}
\label{app:hierarchy-witnesses}
In this appendix, we give separating examples for the inclusion
hierarchies presented in \Cref{sec:class-hierarchy}, first for the
two-dimensional setting, then for the Boolean setting, and finally
for the general setting.

In the following examples, whenever a set
$S\subseteq\mathbb Z^d$ with $d<n$ is used in dimension $n$, we
identify it with its zero extension
\[
S\times\{0\}^{n-d}.
\]
All membership and nonmembership assertions stated below are preserved
under this extension.

We first consider the two-dimensional setting. The examples in
\textup{(a)}--\textup{(g)} witness, in order from bottom to top, the
seven proper inclusions in \Cref{prop:hierarchy-in-2-dim}.

\begin{enumerate}
\setlength{\itemsep}{5pt}
\setlength{\parsep}{0pt}

\item[\textup{(a)}]
The Boolean set
\[
D_{\mathrm B}
=
\{\,(0,0),(0,1),(1,1)\,\}
=
\{\,(x_1,x_2)\in\mathbb Z^2
\mid
0\leq x_1\leq x_2\leq1\,\}
\]
is DC-representable but not SVPI-representable. Indeed, every
SVPI-representable set is the Cartesian product of its one-coordinate
projections, whereas both one-coordinate projections of
$D_{\mathrm B}$ are $\{0,1\}$ and
$D_{\mathrm B}\neq\{0,1\}^2$.

\item[\textup{(b)}]
The Boolean set
\[
\begin{aligned}
U_{\mathrm B}
&=
\{\,(0,1),(1,0),(1,1)\,\}\\
&=
\{\,(x_1,x_2)\in\mathbb Z^2
\mid
0\leq x_1,x_2\leq1,\  x_1+x_2\geq1\,\}.
\end{aligned}
\]
is UTVPI-representable. By
\Cref{cor:DC-fixed-point-characterization}, it is not
DC-representable, since
\[
(0,1),(1,0)\in U_{\mathrm B}
\quad\text{and}\quad
m^-\bigl((0,1),(1,0)\bigr)
=
(0,0)
\notin U_{\mathrm B}.
\]

\item[\textup{(c)}]
Let
\[
T_0
=
\{\,(x_1,x_2)\in\mathbb Z^2
\mid
x_1+2x_2\geq2,\quad 0\leq x_1,x_2\leq2\,\}.
\]
Then $T_0$ is TVPI-representable. By
\Cref{thm:main}, it is not
UTVPI-representable, since
\[
(2,0),(0,1)\in T_0 \ \mbox{ and } \
\avgd\bigl((2,0),(0,1)\bigr)=(1,0)\notin T_0.
\]

\item[\textup{(d)}]
The set
\[
H_0
=
(\mathbb Z\times\{0\})\cup\{(0,1)\}
\]
is hole-free but not closed-hole-free. Indeed,
\[
\conv(H_0)
=
\bigl(\mathbb R\times[0,1)\bigr)\cup\{(0,1)\} \ \mbox{ and } \
\overline{\conv}(H_0)
=
\mathbb R\times[0,1].
\]
Consequently,
\[
\conv(H_0)\cap\mathbb Z^2
=
H_0 \ \mbox{ and } \
\overline{\conv}(H_0)\cap\mathbb Z^2
=
\mathbb Z\times\{0,1\}
\supsetneq
H_0.
\]

\item[\textup{(e)}]
The set
\[
G_0
=
\{\,(0,0),(2,0)\,\}
\]
is median-closed, as is every two-element set, but it is not
hole-free, since
\[
(1,0)\in\conv(G_0)\cap\mathbb Z^2\setminus G_0.
\]

\item[\textup{(f)}]
Let
\[
R_{\mathrm{maj}}
=
\{\,(0,0),(1,2),(2,1)\,\}.
\]
Define $f\colon\mathbb Z^3\to\mathbb Z$ to return the value occurring
at least twice whenever such a value exists, and to return $0$ when
its arguments are pairwise distinct. Then $f$ is a majority operation.
For any triple of points from $R_{\mathrm{maj}}$, its componentwise
image under $f$ is the repeated point if one exists, and is $(0,0)$
otherwise. Hence $R_{\mathrm{maj}}$ is $f$-closed. However, it is not
median-closed, because
\[
\median\bigl((0,0),(1,2),(2,1)\bigr)
=
(1,1)
\notin
R_{\mathrm{maj}}.
\]

\item[\textup{(g)}]
The set
\[
R_{\mathrm{nomaj}}
=
\{\,(0,1),(0,2),(1,2),(2,0)\,\}
\]
is $2$-decomposable, since every subset of $\mathbb Z^2$ is
$2$-decomposable. However, it is not closed under any majority
operation. Suppose, to the contrary, that a majority operation $f$
preserves $R_{\mathrm{nomaj}}$, and set
\[
\alpha=f(0,1,2),
\qquad
\beta=f(1,2,0).
\]
Applying $f$ componentwise to
\[
(2,0),(0,1),(0,2)
\qquad\text{and}\qquad
(0,2),(1,2),(2,0)
\]
gives, respectively,
\[
(0,\alpha)\in R_{\mathrm{nomaj}},
\qquad
(\alpha,2)\in R_{\mathrm{nomaj}},
\]
and hence $\alpha=1$. Similarly, applying $f$ componentwise to
\[
(0,1),(0,2),(2,0)
\qquad\text{and}\qquad
(1,2),(2,0),(0,2)
\]
gives, respectively,
\[
(0,\beta)\in R_{\mathrm{nomaj}},
\qquad
(\beta,2)\in R_{\mathrm{nomaj}},
\]
and hence $\beta=1$. Applying $f$ componentwise to
$(0,1),(1,2),(2,0)$ now gives
\[
(\alpha,\beta)
=
(1,1)
\notin
R_{\mathrm{nomaj}},
\]
contradicting the assumption that $f$ preserves
$R_{\mathrm{nomaj}}$.
\end{enumerate}

We next turn to the Boolean setting. By the zero-extension convention
above, for every $n\geq2$, the Boolean sets $D_{\mathrm B}$ and
$U_{\mathrm B}$ in \textup{(a)} and \textup{(b)} are regarded as
subsets of $\{0,1\}^n$. They witness, respectively, the bottom two
proper inclusions in \Cref{prop:hierarchy-in-Boolean}. For every
$n\geq3$, the following three-dimensional Boolean set, viewed as a
subset of $\{0,1\}^n$ by zero extension, witnesses the remaining
proper inclusion. When $n=2$, the top two classes coincide.

\begin{enumerate}
\setlength{\itemsep}{5pt}
\setlength{\parsep}{0pt}

\item[\textup{(h)}]
Consider the Boolean even-parity set
\[
E_{\oplus}
=
\{\,
x\in\{0,1\}^3
\mid
x_1\oplus x_2\oplus x_3=0
\,\}, 
\]
where $\oplus$ denotes addition modulo $2$. 
Then $E_{\oplus}$ is not $2$-decomposable, since every two-coordinate projection of $E_{\oplus}$ is
$\{0,1\}^2$, and hence the join of these projections  is
$\{0,1\}^3\supsetneq E_{\oplus}$.
Thus, $E_{\oplus}$ belongs to $\mathsf{ALL}$ but not to $\mathsf{2Dec}$.
\end{enumerate}

Finally, we consider the general setting over $\mathbb Z^n$ with
$n\geq3$. By the zero-extension convention, the preceding examples
already witness every proper inclusion in
\Cref{prop:representability-hierarchy} except
\[
\Dmp\subsetneq\mathsf{IC}.
\]
The following example provides the remaining witness.

\begin{enumerate}
\setlength{\itemsep}{5pt}
\setlength{\parsep}{0pt}

\item[\textup{(i)}]
Let
\[
I_0=\{a=(-1,1,1),b=(0,0,1),c=(0,1,0),d=(1,0,0)\}.\] All four points lie in the plane
$x_1+x_2+x_3=1$, and
\[
a+d=b+c.
\]
Hence $a,b,d,c$, in this cyclic order, are the vertices of a
parallelogram. Its diagonal $[b,c]$ divides it into two triangles,
and therefore
\[
\conv(I_0)
=
\conv\{a,b,c\}
\cup
\conv\{b,c,d\}.
\]

Let $u\in\conv(I_0)$, and choose one of the two triangles above
containing $u$. Write $u$ as a convex combination of the three
vertices of this triangle. Since, in each coordinate, the values of
these three vertices differ by at most one, every vertex having a
positive coefficient in this convex combination belongs to $N(u)$.
Consequently,
\[
u\in\conv\bigl(I_0\cap N(u)\bigr).
\]
It follows directly from the definition that $I_0$
is integrally convex.

On the other hand,
\[
\avgd\bigl((-1,1,1),(1,0,0)\bigr)
=
(0,1,1)
\notin
I_0.
\]
Thus, $I_0$ is not $\avgd$-closed, and hence
\[
I_0\in\mathsf{IC}\setminus\Dmp.
\]
This proves the remaining proper inclusion.

\end{enumerate}

\end{document}